\documentclass[a4paper,USenglish,cleveref, autoref]{lipics-v2019}

\usepackage{booktabs} % For formal tables
\usepackage{amsfonts,amssymb,amsmath,latexsym,ae,aecompl,bbm,algorithm,amsthm}
\usepackage[noend]{algpseudocode}
\usepackage{graphicx}
\usepackage{color}
\usepackage{subcaption}
\usepackage{cite}

\newtheorem{observation}{Observation}

\makeatletter
\newtheorem*{rep@theorem}{\rep@title}
\newcommand{\newreptheorem}[2]{%
	\newenvironment{rep#1}[1]{%
		\def\rep@title{#2 \ref{##1}}%
		\begin{rep@theorem}}%
		{\end{rep@theorem}}}
\makeatother

\newreptheorem{theorem}{Theorem}
\newreptheorem{lemma}{Lemma}

\newcommand{\blue}[1]{\textcolor{blue}{#1}}
\newcommand{\magenta}[1]{\textcolor{magenta}{#1}}
\newcommand{\cyan}[1]{\textcolor{cyan}{#1}}
\newcommand{\green}[1]{\textcolor{green}{#1}}
\def\ShowComment{True}

\ifdefined\ShowComment
\def\billy#1{\marginpar{$\leftarrow$\fbox{B}}\footnote{$\Rightarrow$~{\sf #1 \blue{--Billy}}}}
\else
\def\billy#1{}
\fi

\ifdefined\ShowComment
\def\ron#1{\marginpar{$\leftarrow$\fbox{R}}\footnote{$\Rightarrow$~{\sf #1 \magenta{--Ron}}}}
\else
\def\ron#1{}
\fi

\ifdefined\ShowComment
\def\shay#1{\marginpar{$\leftarrow$\fbox{S}}\footnote{$\Rightarrow$~{\sf #1 \cyan{--Shay}}}}
\else
\def\shay#1{}
\fi

\ifdefined\ShowComment
\def\yuval#1{\marginpar{$\leftarrow$\fbox{Y}}\footnote{$\Rightarrow$~{\sf #1 \green{--Yuval}}}}
\else
\def\yuval#1{}
\fi

\title{%{\Huge Appendix: Full Paper} \protect\\ ~\protect\\
 Deterministic Leader Election in Programmable Matter}

\author{Yuval Emek}{Faculty of Industrial Engineering and Management, Technion - IIT, Haifa, Israel}{yemek@technion.ac.il}{}{The work of Y. Emek was supported in part by an Israeli Science Foundation grant number 1016/17.}

\author{Shay Kutten}{Faculty of Industrial Engineering and Management, Technion - IIT, Haifa, Israel}{kutten@ie.technion.ac.il}{}{The work of this author was supported in part by a grant from the Bi-national Science Foundation.}

\author{Ron Lavi}{Faculty of Industrial Engineering and Management, Technion - IIT, Haifa, Israel}{ronlavi@ie.technion.ac.il}{}{This research was supported by the ISF-NSFC joint research program (grant No. 2560/17).}

\author{William K. Moses Jr.\footnote{Corresponding author.}}{Faculty of Industrial Engineering and Management, Technion - IIT, Haifa, Israel}{wkmjr3@gmail.com}{0000-0002-4533-7593}{The work of this author was supported in part by a grant from the Israeli Ministry of Science.}

\titlerunning{Deterministic Leader Election in Programmable Matter}

\authorrunning{Y. Emek et al.}

\Copyright{Yuval Emek, Shay Kutten, Ron Lavi, and William K. Moses Jr.}



\nolinenumbers %uncomment to disable line numbering

\hideLIPIcs  %uncomment to remove references to LIPIcs series (logo, DOI, ...), e.g. when preparing a pre-final version to be uploaded to arXiv or another public repository

\EventEditors{xxx}
\EventNoEds{0}
\EventLongTitle{xxx}
\EventShortTitle{xxx}
\EventAcronym{xxx}
\EventYear{2019}
\EventDate{xxx}
\EventLocation{xxx}
\EventLogo{}
\SeriesVolume{0}
\ArticleNo{0}
\begin{document}

\maketitle

\begin{abstract}
% !TEX root = main.tex
 %please leave the above line untouched. Thanks, Billy.

Addressing a fundamental problem in programmable matter, we present the first deterministic algorithm to elect a unique leader in a system of connected amoebots assuming only that amoebots are initially contracted. Previous algorithms either used randomization, made various assumptions (shapes with no holes, or known shared chirality), or elected several co-leaders in some cases. 

Some of the building blocks we introduce in constructing the algorithm are of interest by themselves, especially the procedure we present for reaching common chirality among the amoebots. Given the leader election and the chirality agreement building block, it is known that various tasks in programmable matter can be performed or improved.

The main idea of the new algorithm is the usage of the ability of the amoebots to move, which previous leader election algorithms have not used.
\end{abstract}

%\textbf{Keywords:} programmable matter,
%geometric amoebot model,
%leader election

%----------------------------------------------------------------------------------------------

\section{Introduction}
\label{sec:intro}
% !TEX root = main.tex
 %please leave the above line untouched. Thanks, Billy.

The notion of programmable matter was introduced by Toffoli and Margolus in \cite{TM91}. The main purpose was to provide a conceptual way to model matter that could change its physical properties in some programmable way. This model envisioned matter as a group of individual particles interacting with each other subject to some constraints. Individually, each particle was a single computational entity. But when taken together, they produced matter that could be programmed to act in specific desired ways. In the context of this framework, it becomes possible to address what computational problems can be solved by such groups of particles and how efficiently they can be solved.

The amoebot model was first proposed by Derakhshandeh et al.~\cite{DDGRSS14,DHRS19} as a possible abstraction for computing at the micro and nano scales. Amoebots represent finite memory mobile agents that move in a manner inspired by amoeba and must stay connected at all times. This model was quickly adopted by the community and much work was done on several problems such as coating of materials~
\cite{DGRSST14,DGRSS17b,DDGPRSS18}, bridge building~\cite{ACDRR18}, shape formation~
\cite{DGRSS15,CDRR16,DGRSS16,DFSVY18,DGHKSR18}, and shape recovery~\cite{DFPSV18}. An important primitive used often is the election of a unique leader that subsequently breaks symmetry and coordinates the remaining amoebots.

Leader election has been studied for a wide variety of assumptions. %Previous solutions required the system to either have no holes, have common chirality, use randomness, or permit more than one leader to be elected. 
However, being such a fundamental primitive, an important goal is to develop a robust leader election algorithm that requires as few assumptions as possible.

\subsection{Amoebot Model}
The \emph{amoebot model} considers several \emph{amoebots} (also called \emph{particles}) located in the nodes of a graph $G$. In this paper, we consider the \emph{geometric amoebot model} where $G$ is assumed to be the infinite regular triangular grid naturally embedded in the plane as illustrated in Figure~\ref{fig:triangular-grid}.

\begin{figure}
	\includegraphics[page=1,height=2.2in]{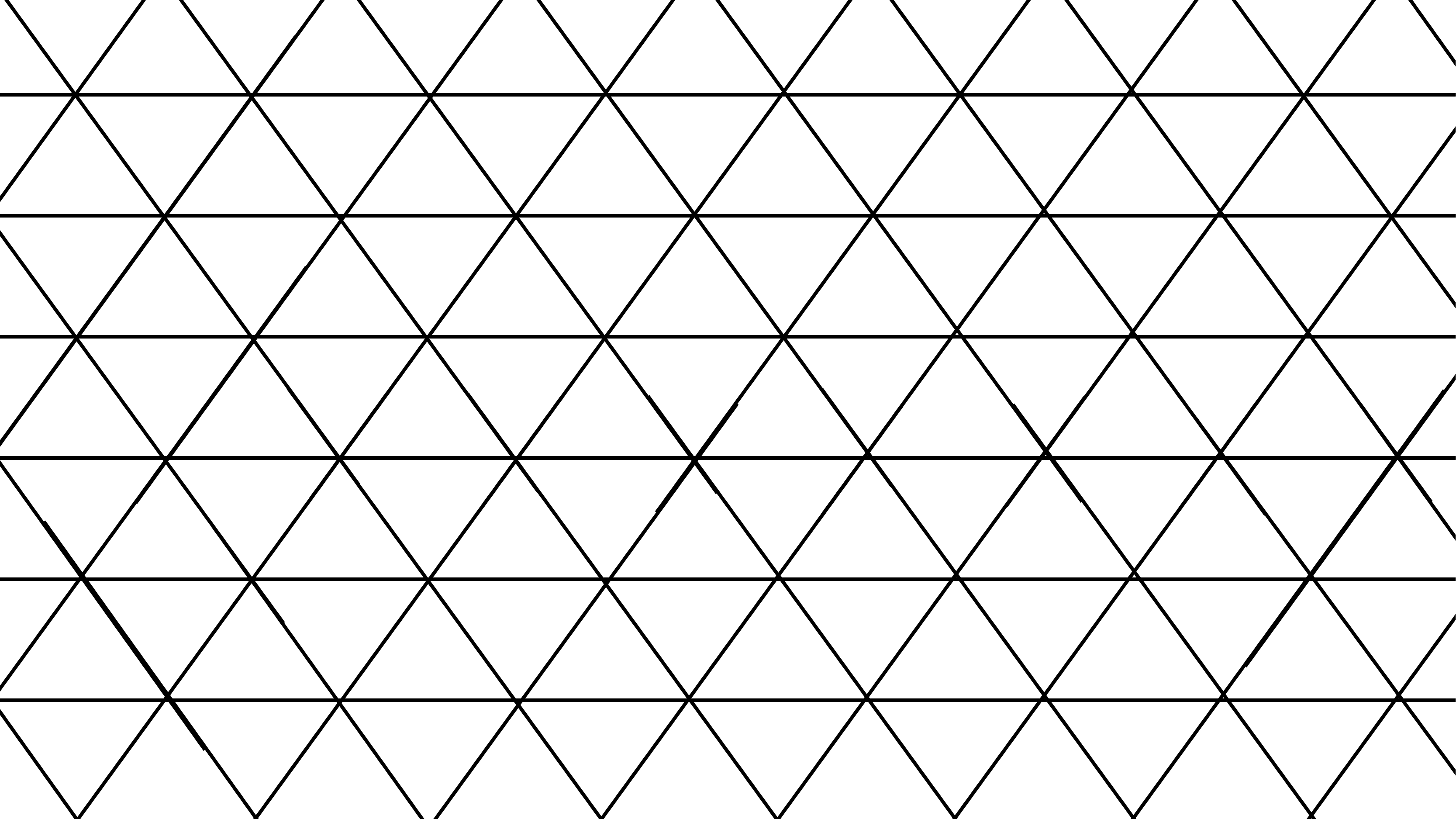}
	\caption{A cross-section of the natural planar embedding of an infinite regular triangular grid. A node of the graph is indicated by the intersection of the lines of the grid. A particle present at any node may move in any of the six directions indicated by lines going out from it.} \label{fig:triangular-grid}
\end{figure} 

Each particle has constant size memory and following the common practice in the amoebot model literature (see, e.g.,~\cite{DGSBRS15}), we assume that it is well-initialized prior to the start of the algorithm. In the \emph{leader election (LE)} problem, the particles have one of three \emph{(LE) statuses}: \emph{candidate}, \emph{leader}, or \emph{unelected}, denoted by \textbf{C}, \textbf{L}, and \textbf{U}, respectively. Initially, each particle is a possible leader and has status \textbf{C} and will permanently change its status to either \textbf{U} or \textbf{L} by the end of the algorithm. When we say that a node has one of the three statuses, we mean that the particle occupying that node has that status. Each particle is either \emph{contracted} or \emph{expanded} depending on whether it occupies one node or two adjacent nodes of the grid, respectively. 
%For each node that a particle occupies, it maintains six port numbers indicating each of the six edges leading to neighboring nodes in the graph. These port numbers are in an increasing sequence from $0$ to $5$, starting at some arbitrary edge. One such port numbering is illustrated in Figure~\ref{fig:clockwise-ordering-ports}. The \textit{chirality} of the particle refers to whether this sequence is oriented in a clockwise manner or an anticlockwise manner. 

The particles are classified according to their \emph{chirality} as \emph{clockwise (CW) particles} or \emph{counter-clockwise (CCW) particles} so that a CW (resp., CCW) particle numbers the \emph{ports} corresponding to the $6$ incident edges in (each of) the node(s) it occupies from $0$ to $5$ in increasing CW (resp., CCW) order, where the edge from which this numbering starts is chosen arbitrarily (refer to Figure~\ref{fig:cw-ccw-ordering-ports} for an illustration). We assume that the particle chirality classification is determined by a malicious adversary and that, initially, the particles are not aware of their own chirality nor are they aware of the chiralities of their adjacent particles. The \emph{degree} of a particle $P$ is the number of adjacent particles to $P$. 
 
The \emph{configuration} of particles at a given instant of time comprises of the subgraph induced on the grid by occupied nodes, the specific node(s) occupied by each particle, and each particle's internal memory. %the internal state of each particle, and the port ordering(s) for each particle.
%A \emph{configuration} of $n$ particles at a given instant of time $t$ is the tuple $ \langle G_t, C_t, M_t \rangle$ where $G_t$ is the subgraph induced by nodes occupied by the particles at time $t$ where each node is labeled with the occupying particle, $C_t$ is an $n$ dimensional vector representing whether each of the $n$ particles is in an expanded or contracted state at time $t$, and $M_t$ is an $n$ dimensional vector representing the state (i.e., internal memory) of each of the $n$ particles at time $t$. 
The configuration is said to be \emph{contracted} if all particles are in a contracted state. We follow the common practice (see~\cite{DGRSS17a}) and assume that the particles are, initially, in a contracted configuration. The algorithm terminates in a contracted configuration.

%\begin{figure}
%	\includegraphics[page=2,height=2.2in]{pics.pdf}
%	\caption{One possible clockwise ordering of ports.} \label{fig:clockwise-ordering-ports}
%\end{figure} 

%\begin{figure}
%	\includegraphics[page=14,height=2.2in]{pics.pdf}
%	\caption{Three CW particles with different port orderings.} \label{fig:clockwise-ordering-ports}
%\end{figure} 

\begin{figure}
	\includegraphics[page=29,height=2.2in]{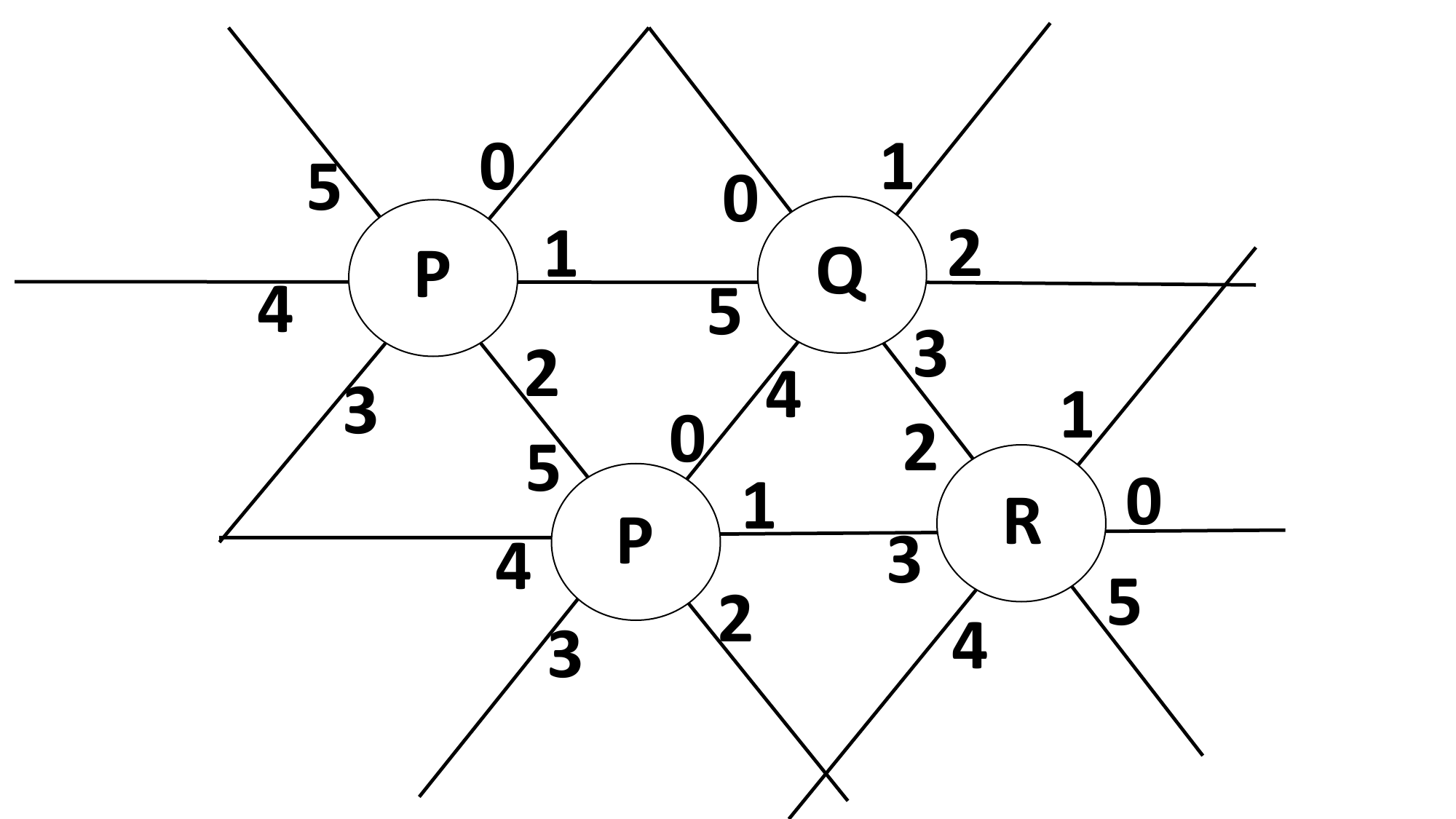}
	\caption{Port orderings of four nodes occupied by two CW particles $P$ and $Q$ and one CCW particle $R$.} \label{fig:cw-ccw-ordering-ports}
\end{figure} 

The subgraph $G(t)$ induced on the grid by the nodes occupied by particles at time $t$ is referred to as a \emph{shape}. Following the common practice in the amoebot model literature (see~\cite{DDGRSS14}), it is required that the shape is connected at all times. Since the shape is a (finite) planar graph associated with a planar embedding, it partitions the plane into \emph{faces} (see~\cite{D05}), where exactly one of them is unbounded, a.k.a. the \emph{outer} face. The occupied nodes adjacent to the outer face are said to form the \emph{outer boundary} of the shape. An inner face that includes at least one unoccupied node is called a \emph{hole} in the shape. The occupied nodes adjacent to a hole are said to form an \emph{inner boundary}. An example of a shape with holes illustrating which nodes are boundary nodes is given in Figure~\ref{fig:amoebots-two-holes}. The \emph{length} of a boundary $B$, denoted $L_B$, is the number of nodes on that boundary. Define $L_{\max} = \max \limits_{B} L_B$.

\emph{Boundary particles} are particles that lie on either an inner or outer boundary. A {\em local boundary} of a particle is an interval $i, i+1, \ldots, i+j \mod 6$ of its ports that lead to unoccupied grid nodes.
 Note that a contracted particle may have up to three local boundaries, each a part of some boundary of the shape. However, all three may be parts of the same boundary of the shape. Henceforth, 
 we use only the term ``boundary'' even for local boundaries, when the context makes the usage clear.
 
    A \emph{bridge particle} is a {\em contracted} boundary particle occupying a node $b$ lying on $i$ boundaries $1 \leq i \leq 3$, each of which is the outer boundary, and having $i$ occupied adjacent nodes in the grid. A \emph{semi-bridge particle} is a {\em contracted} boundary particle occupying a node $b$ lying on $2$ outer boundaries, and having $3$ or $4$ occupied adjacent nodes in the grid. If $b$ is occupied by a bridge or semi-bridge particle, $c$ is an adjacent occupied node, and both sides of the edge $(b,c)$ are the outer boundary, then edge $(b, c)$ is called a \emph{bridge edge}.\footnote{Note that a semi-bridge particle may have 3 adjacent occupied nodes and lie on 2 outer boundaries and 1 inner boundary. In this case, the particle will only have 1 bridge edge.}%\billy{The reason I keep writing ``contracted" everywhere is that I want to leave myself room for extra definitions in case some of the algorithms make claims on boundaries of configurations where some of the boundary particles are expanded.} 
An example of bridge particles and semi-bridge particles with bridge edges is illustrated in Figure~\ref{fig:bridge-semibridge-particles}.

\begin{figure}
	\includegraphics[page=3,height=2.2in]{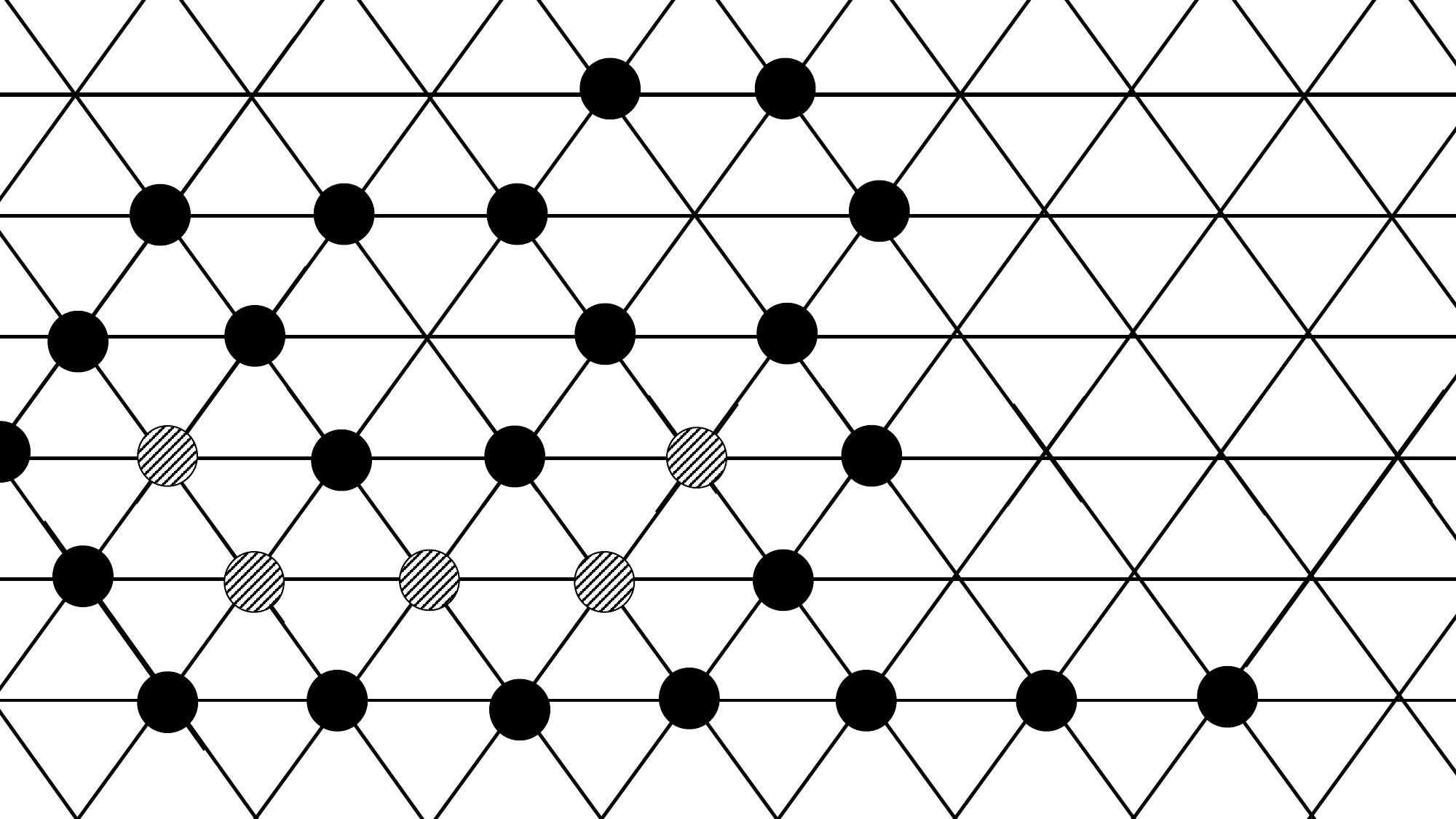}
	\caption{A shape with two holes. Boundary nodes are filled and non-boundary nodes are patterned.} \label{fig:amoebots-two-holes}
\end{figure} 

\begin{figure}
	\includegraphics[page=33,height=2.2in]{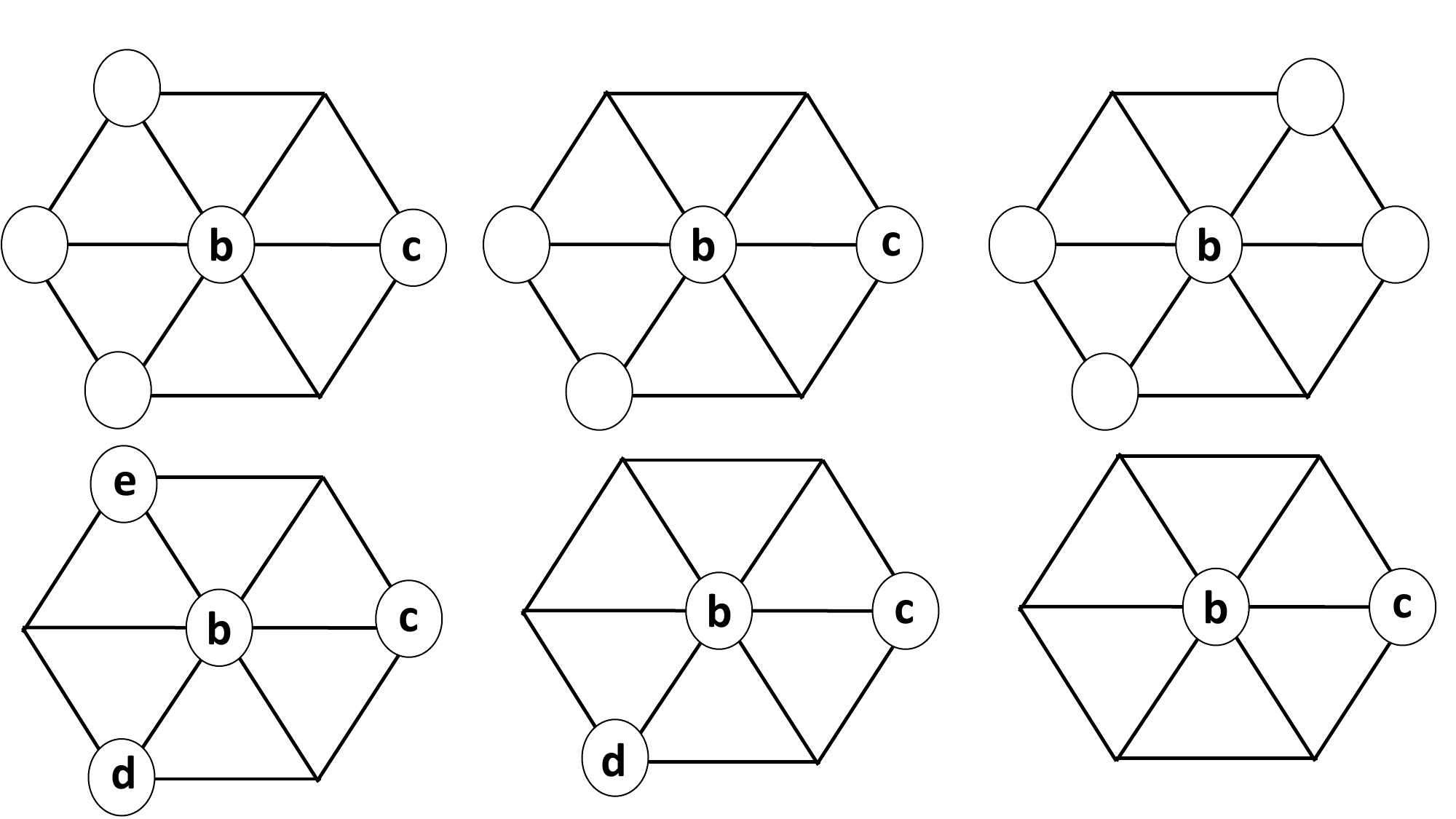}
	\caption{Let node $b$ be occupied by a contracted particle $P$. The top row and bottom row contain examples of when $P$ is a semi-bridge particle and a bridge particle, respectively. The edges $(b,c)$, $(b, d)$, and $(b,e)$ are bridge edges.} \label{fig:bridge-semibridge-particles}
\end{figure}

%However there may exist \textit{holes} in this system. A hole is a connected subgraph of one or more nodes surrounded by particles. The connected series of particles which surround such a hole, i.e. occupy nodes adjacent to the empty nodes, are said to form an \textit{inner boundary} around that hole. A similar notion of \textit{outer boundary} is used to describe the connected series of s which occupy the outermost positions of the system and thus are adjacent to empty nodes which do not form holes. We call particles that are present on a boundary, either inner or outer, \textit{boundary particles}. A boundary particle, which is contracted, may be present on one, two or three boundaries. An illustration of a system of particles with holes and contracted boundary particles is given in Figure~\ref{fig:amoebots-two-holes}. For a given contracted boundary particle, if all adjacent particles are also boundary particles and all boundaries that particle is on are outer boundaries, then we call that particle an \textit{contracted bridge particle}. 

For a boundary node $b$ occupied by particle $P$ with chirality $C$ and lying on boundary $B$, define $b$'s \textit{predecessor} node $a$ and \textit{successor} node $c$ w.r.t.\ $B$ and $C$ as the previous occupied node and the next occupied node along $B$ according to $C$, respectively (refer to Figure~\ref{fig:predecessor-successor} for an illustration).\footnote{Throughout, we use w.r.t.\ to abbreviate `with respect to'.} Note that node $b$ admits such predecessor $a$ and successor $c$ for each boundary $b$ lies on. 

The \textit{boundary count} of $b$ w.r.t.\ $B$ and $C$ measures the deviation of the line segment formed by $b$ and its successor from the line segment formed by $b$'s predecessor and $b$ w.r.t.\ $B$ taking $C$ into account. More formally, the boundary count of $b$ w.r.t.\ $B$ is a function of $C$ and the angle $\angle abc$ that takes on one of the values $-1, 0, 1, 2,$ or $3$ (as illustrated in Figure~\ref{fig:boundary-count}).\footnote{Note that it is not possible to have a node with boundary count -2 or -3 w.r.t.\ some boundary.~\\ The boundary count and its application to calculating the count of a segment, to be defined later, is similar to how Bazzi and Briones~\cite{BB18} use vertex labeling in deciding the count of a segment in their paper. The actual measurement of the boundary count is similar to how Derakhshandeh et al.~\cite{DGSBRS15} measure the angles between the direction a token enters and exits an agent.} Let $i$ be the unique integer that satisfies $\angle abc = 180^\circ - i*60^\circ$. Let $x$ and $y$ be the port numbers of $b$ corresponding to edges $(b,a)$ and $(b,c)$, respectively. If $(x - y) \mod 6 = 4$, then the boundary count of $b$ w.r.t.\ $B$ is $-i$, else it is $i$. In the current paper, when the boundary referred to is clear from context, it is not mentioned when giving the boundary count for a node.

Consider an occupied node $b$ on boundary $B$ with boundary count $w$. The following definitions for $b$ are all w.r.t.\ $B$. When $w= -1, 1, 2,$ or $3$, $b$ is a \emph{vertex}.\footnote{Notice that the angle bisector of a vertex with boundary count 1 or -1 overlaps with a line of the triangular grid.} When $w=2$, $b$ is a \emph{sharp} vertex.  Vertex $b$ is \emph{concave} when $w=-1$ and \emph{convex} when $w=1$ or $2$. A shape whose outer boundary vertices are all convex w.r.t.\ the outer boundary is a \emph{convex polygon}.  A shape consists of \emph{two (or more) simple convex polygons sharing the same semi-bridge particle(s) $P$ (, $Q$, $R$, etc.)} when (i) $P$ (, $Q$, $R$, etc.) has no adjacent bridge edges, (ii) the shape is disconnected by removing $P$ (, $Q$, $R$, etc.), and (iii) all vertices other than those occupied by $P$ (, $Q$, $R$, etc.) are convex vertices.  Notice that the definition of convex polygon relates only to its outer boundary nodes. Specifically, no assumptions are made on the presence of holes within the shape.

\begin{figure}
	\includegraphics[page=4,height=2.2in]{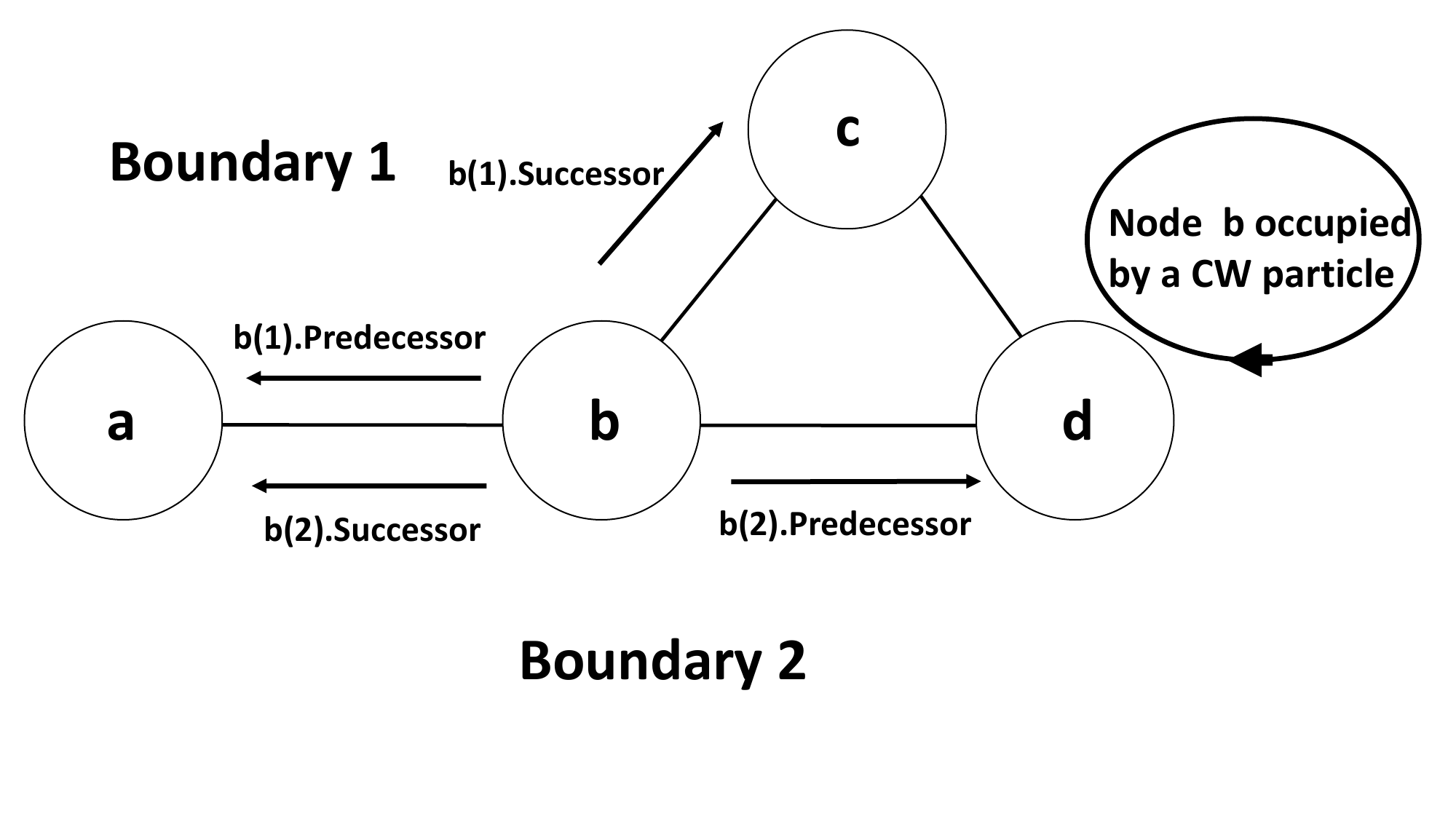}
	\caption{Cross-section of a shape with four boundary nodes $a$, $b$, $c$, and $d$, where $b$ is occupied by a CW particle. The predecessor and successor nodes of $b$ are shown with respect to the two boundaries $b$ lies on.} \label{fig:predecessor-successor}
\end{figure} 

\begin{figure}
	\includegraphics[page=16,height=2.2in]{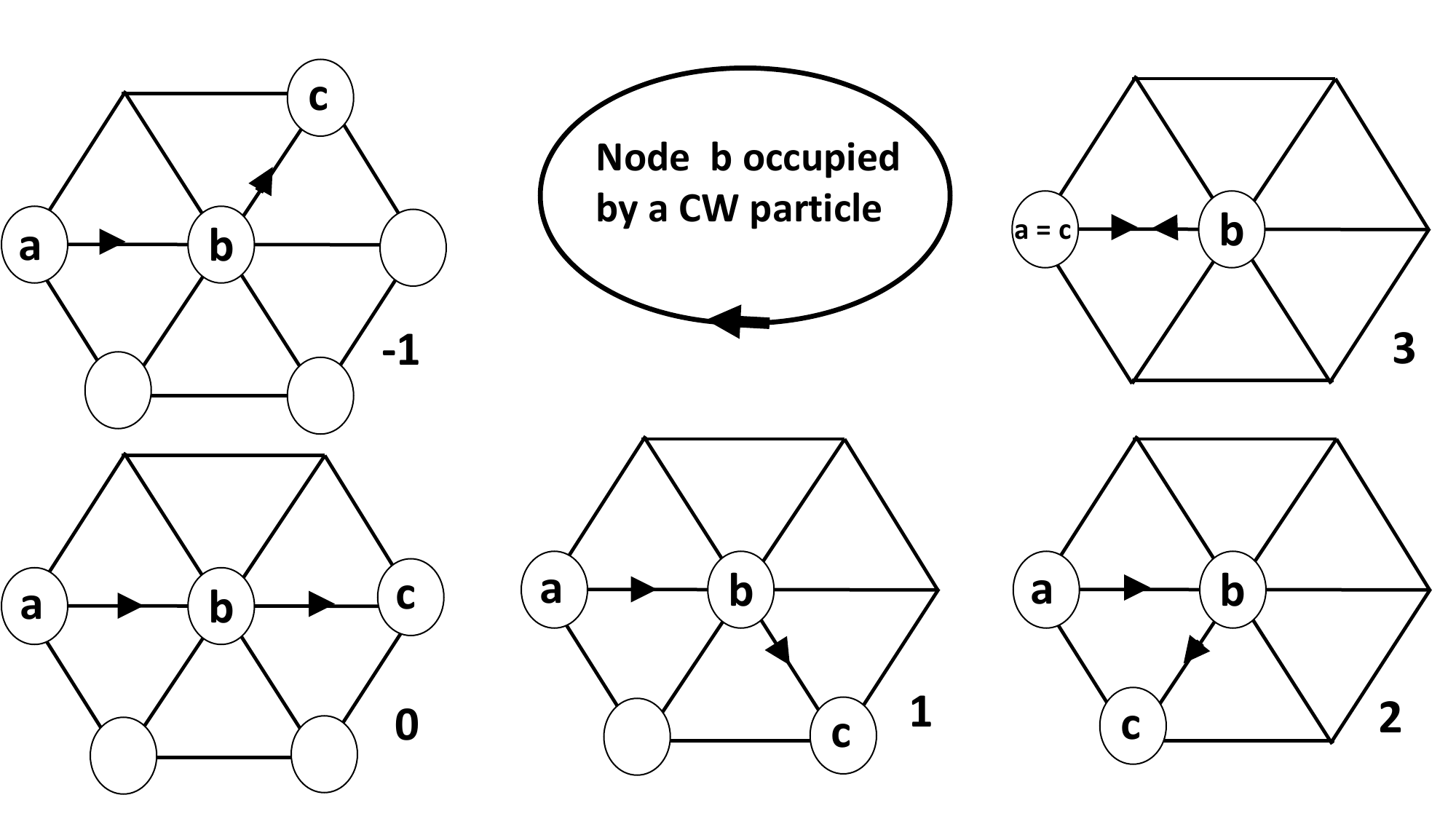}
	\caption{The five possible boundary count values for a node $b$, lying on exactly one boundary, with predecessor node $a$ and successor node $c$ w.r.t.\ that boundary.% Dashed arrows indicate the deviation of the line $by$ from the line $xb$.
	} \label{fig:boundary-count}
\end{figure} 

\subsubsection{Scheduling of Particles}
The particles are said to be \emph{activated} by an asynchronous scheduler such that the time interval between any two activations of the same particle is finite. Each activation of a particle $P$ consists of 3 stages.
\begin{enumerate}
	\item $P$ reads the memories of adjacent particles.
	\item $P$ performs some local computation and may update its own memory and/or the memories of its neighboring particles.
	\item $P$ may perform a movement operation.
\end{enumerate}
Each activation of a particle is atomic in that once activated, a particle will complete all 3 stages before the next particle is activated. One \emph{asynchronous round} is completed when each particle is activated at least once. For convenience, when we say that a node $b$ performs an action it means that the particle occupying $b$ performs that action when activated. Notice that a particle may occupy two nodes and maintain different state information for each node it occupies.

\subsubsection{Movement of Particles}
As mentioned already, particles occupy either one node or two nodes of the graph. Particles  only move to adjacent nodes in the graph so long as connectivity of the shape formed by the particles is not broken. The particles move via a series of \textit{expansions} and \textit{contractions}. A particle performs an expansion (resp., contraction) only when it is in a contracted (resp., expanded) state. 

Consider a contracted particle $P$ occupying a node $a$ and an expanded particle $Q$ occupying nodes $b$ and $c$. $P$ expands by extending itself to an adjacent unoccupied node such that $P$ now occupies two nodes. Furthermore, $P$ maintains a distinct notion of \emph{head} node and \emph{tail} node referring to the nodes which $P$ expanded into and expanded from, respectively. $Q$ contracts by moving itself completely into either $b$ or $c$. 
If $a$ and $b$ are adjacent, the model allows $P$, when activated, to expand into $b$ and force $Q$ to contract into $c$. $P$ is said to have \emph{pushed} $Q$ into $c$. Similarly, the model allows $Q$, when activated, to \emph{pull} $P$ into $b$ by contracting into $c$ and forcing $P$ to expand into $b$.

\subsubsection{Inter-particle Communication}
An activated particle $P$ communicates with the set of particles $\mathcal{S}$ located in nodes adjacent to $P$'s node(s) by reading and/or writing to their memories. For convenience, say that $P$ \emph{receives messages} from particles that were earlier activated and wrote into $P$'s memory and $P$ \emph{sends messages} to particles in $\mathcal{S}$ that it writes into the memories of.\footnote{Note that $P$ allocates memory for each of its ports and that is the memory that can be modified by adjacent particles. In other words, when $P$ receives a message, it knows through which port the message was sent and by extension which particle sent it.}

\subsubsection{Problem Statement}
Consider a contracted configuration of particles forming a connected shape. Particles may not have common chirality initially and the shape may have holes. Design a terminating algorithm to be run by each particle such that at the end, one particle has status \textbf{L} and the remaining particles have status \textbf{U}. 

\subsection{Our Contributions}
The current paper presents the first deterministic protocol that elects exactly one leader on any contracted configuration. We first assume that all particles have common chirality and Section~\ref{sec:chir-agreement} explains how to remove this assumption. For a comparison of the result to known results, see Table~\ref{table:results-comparison}.

Before presenting the main algorithm, in Section~\ref{sec:building-blocks}, four building blocks are developed that may be of interest on their own. A fifth building block (chirality agreement) is presented in Section~\ref{sec:chir-agreement}. The tools are a maximal independent set (MIS) selection protocol, a boundary detection protocol, a leader election protocol on a convex polygon without sharp vertices, and a leader election protocol on a spanning tree. 

The main protocol, $\mathtt{Leader-Election-By-Moving}$, is presented in Section~\ref{sec:leader-election-by-moving}. It is a 6 stage deterministic algorithm which utilizes the tools mentioned in the previous paragraph as well as other techniques to finally elect exactly one leader.

The assumption that the particles have a common chirality is removed in Section~\ref{sec:chir-agreement}, i.e. a procedure is presented that guarantees all particles will have common chirality when the procedure terminates.

\begin{table*}[ht]
	\caption{Table comparing the result on leader election to those of previous papers. ``No holes" refers to whether the algorithm requires the graph to have no holes initially or not. ``Multiple leaders" refers to whether the leader election algorithm may output multiple leaders in certain cases or always outputs a unique leader. The length of the largest boundary in the initial configuration is $L_{\max}$. The length of the outer boundary in the initial configuration is $L$. The number of particles in the configuration is denoted by $n$. The terms $r$ and $mtree$ are unique to paper~\cite{GAMT18}.}
	\centering \vspace{1em}
		\resizebox{1.0\columnwidth}{!}{%
	\begin{tabular}{|c|c|c|c|c|c|}
		\hline
		Paper & Common & Randomness & No holes & Multiple  & Running time \\
		& chirality &  &  & leaders & \\
		\hline
		\hline
		\cite{DGSBRS15} & Yes & Yes & No & No & $O(L_{\max})$ rounds on expectation \\
		\hline
		\cite{DGRSS17a} & Yes & Yes & No & No & $O(L)$ rounds with high probability  \\
		\hline
		\cite{BB18} & Yes & No & No & Yes & Not analyzed in paper \\
		\hline
		\cite{DFSVY18} & No & No & Yes & Yes & $O(n)$ rounds \\
		\hline
		\cite{GAMT18} & Yes & No & Yes & No & $2(r + mtree + 1)$ rounds \\
		\hline
		Current Paper & No & No & No & No & $O(Ln^2)$ rounds\\
		\hline
	\end{tabular}
		}
	\label{table:results-comparison}
\end{table*}

\subsection{Related Work}
Derakhshandeh et al.~\cite{DGSBRS15} were the first to study leader election in the amoebot model. Assuming common chirality initially, they proposed a randomized algorithm to achieve leader election in $O(L_{\max})$ rounds on expectation, where $L_{\max}$ was the length of the largest boundary in the shape. Derakhshandeh et al.~\cite{DGRSS17a} also assumed common chirality initially but improved upon the result by presenting a randomized algorithm that elected a unique leader in $O(L)$ rounds, where $L$ is the length of the outer boundary of the shape. A deterministic algorithm was presented by Di Luna et al.~\cite{DFSVY18} to elect a leader and obtain common chirality for the natural special case that the shape did not contain holes; (a constant number of) multiple leaders could be elected in some cases. Di Luna et al.~\cite{DFSVY18} then use the leader(s) to perform shape transformation. Bazzi and Briones~\cite{BB18} assumed common chirality and presented a brief announcement outlining an algorithm to deterministically elect a leader and in some cases, a constant number of multiple leaders, even when holes were present. The current paper adapts and uses a sub-routine from their paper. Gastineau et al.~\cite{GAMT18} presented a leader election algorithm that assumed common chirality and a connected hole-free shape of particles and elected a leader with the additional property of assigning an identifier to each particle that is unique within a radius of $k$ particles. Their result holds for triangular, square, and king grids, while the current paper focuses on the triangular grid alone.\footnote{It is possible to adapt the current paper's leader election algorithm to run on king grids.}

The type of asynchronous scheduler used affects the leader election results. Typically in the literature~\cite{DGSBRS15,DGRSS17a,GAMT18}, the scheduler provides conflict resolution mechanisms for movement and communication such that particle activations can be analyzed sequentially, i.e., the activation of each particle is atomic. However, when a scheduler is allowed to schedule such particles simultaneously~\cite{BB18,DFSVY18}, it becomes impossible to elect deterministically a unique leader in some cases.

\subsection{Technical Challenges and Ideas}
Multiple ideas are combined here in order to address different cases. Consider, for example, a polygon with a hole.
One approach in previous algorithms, assuming no holes, was to remove (from being candidates) boundary nodes repeatedly until only one remains. In the case of one hole (addressed by one of our subroutines), the present algorithm utilizes the ability of particles to move. Intuitively, they may move (eventually) to the center of the polygon, and
the particle reaching the center first is the elected one. (Thanks to the sequential scheduler, only one can reach a certain node first).

This, of course, requires our algorithm to perform various maneuvers, to identify the center and to make sure no additional holes remain. In particular, particles have to identify the outer boundary, move outward in order to gain a symmetric shape, and then move inward together so no additional holes are formed. Since multiple polygons may be moving at the same time, two polygons may ``collide'' and not manage to finish the maneuver. There, we use the idea of reset, to restart the algorithm for the new shape. We managed to upper bound the number of such resets.

Because of the existence of bridge (and semi-bridge particles), solving for a single simple polygon is not enough. For example, consider the case that the shape is a long line (possibly connecting simple polygons). Here,
we use the fact that there exists only one outer face (borrowing its detection from the algorithm of \cite{BB18}, with some necessary adaptations).
The partial leaders of the simple polygons cooperate to define a tree that spans the simple polygons. Final leader election is then performed over the tree.

The assumption of common chirality is used throughout the paper. To remove this assumption and have particles agree on chirality, we use again the detection of the outer face. The particles on the outer boundary agree on chirality (this turned out to be easier for us than agreeing on a leader among them, using the local symmetry breaking provided by the scheduler). Then, the outer boundary particles coordinate and propagate this shared chirality to the other particles within the shape.

\subsection{Organization of the Paper}
The rest of this paper is divided as follows. %Section~\ref{sec:prelims} contains further technical preliminaries. 
In Section~\ref{sec:building-blocks}, four building blocks are presented that are subsequently used in the main protocol. Section~\ref{sec:leader-election-by-moving} contains the main new leader election protocol. Section~\ref{sec:chir-agreement} contains a new procedure for obtaining common chirality of particles in a configuration. Finally, some conclusions and possible future directions of work are discussed in Section~\ref{sec:conc}.

%----------------------------------------------------------------------------------------------

%\section{Technical Preliminaries}
%\label{sec:prelims}
%\input{prelims.tex}

%----------------------------------------------------------------------------------------------

\section{Building Blocks}
\label{sec:building-blocks}
% !TEX root = main.tex
 %please leave the above line untouched. Thanks, Billy.

In this section, we present the four building blocks in detail. 

We first give several definitions. Each boundary particle $P$ maintains a binary flag $seg\_head$ in each boundary node $b$ that $P$ occupies w.r.t.\ each boundary that $b$ lies on. For convenience, when $P$ is a contracted particle, we simply say that $P$ maintains such a flag for each boundary it lies on. When $P$ occupies a boundary node and has $seg\_head = true$ for that boundary, we say $P$ is a \emph{seg-head} for that boundary. Consider two seg-heads $P_1$ and $P_2$ on boundary $B$, occupying nodes $b_1$ and $b_2$ with predecessor nodes $c_1$ and $c_2$ and successor nodes $d_1$ and $d_2$ w.r.t.\ $B$, respectively. $P_2$ is the \emph{previous (resp., next) seg-head} before (resp., after) $P_1$ iff the particles in successor (resp., predecessor) nodes from $b_2$ to $b_1$ w.r.t.\ $B$ (excluding $b_1$ and $b_2$) have $seg\_head$ set to $false$. 

Let $P_2$ be the next seg-head after $P_1$ w.r.t.\ $B$. $P_1$'s \emph{segment} is the sequence of successor nodes from $b_1$ to $c_2$ with \emph{head} $b_1$ and \emph{tail} $c_2$. It is said that $P_1$'s segment is before $P_2$'s segment or $P_2$'s segment is after $P_1$'s segment w.r.t.\ $B$. The \emph{count} of a segment w.r.t.\ $B$, stored by our procedures in the segment's head, is the sum of the boundary counts of its constituent vertices w.r.t.\ $B$. For the sake of convenience, when referring to a procedure/action initiated by the head of a segment involving the participation of the particles in that segment, we just say that a segment runs the procedure/performs the action. It is important to note that a particle $P$ may participate in multiple segments simultaneously (one per boundary $P$ lies on). The algorithm needs to be careful to prevent contradicting actions of such segments (for example, preventing one segment from expanding $P$ into one node while another segment is trying to expand $P$ into a different node).

 The lexicographic comparison of two segments $s_1$ and $s_2$ consists of comparing the boundary counts of their nodes from head to tail. If $s_1$ and $s_2$ are of the same length and have the same boundary counts, then they are said to be lexicographically equal ($s_1 \equiv s_2$). Otherwise, let the position within the segment that $s_1$ and $s_2$ differ in the boundary counts be $x$. If the node of $s_1$ at position $x$ has a boundary count less than that of $s_2$ or the size of $s_1$ is $<x$, then $s_1$ is lexicographically lower than $s_2$  ($s_1 < s_2$). Else $s_1$ is lexicographically higher than $s_2$ ($s_1 > s_2$).

\subsection{MIS Selection}
\label{subsec:mis-selection}
% !TEX root = main.tex
 %please leave the above line untouched. Thanks, Billy.

This tool is called $\mathtt{MIS-Selection}$ and is used as a procedure in our main algorithm. The following trivial observation already breaks with impossibility results in other models when the scheduler is not asynchronous.\footnote{In particular, this procedure selects a leader when 2 or 3 mutually adjacent particles participate.}

\alglanguage{pseudocode}
\begin{algorithm}
\caption{MIS-Selection, run by each particle $P$}
\label{prot:MIS-Selection}
\begin{algorithmic}[1]	
	\If{no other adjacent particle has joined the MIS}
		\State $P$ joins the MIS
	\Else
		\State $P$ does not join the MIS
	\EndIf
\Statex
\end{algorithmic}
\end{algorithm}

\begin{observation}\label{obs:mis-selection}
When run by particles, procedure $\mathtt{MIS-Selection}$ deterministically computes an MIS in one round.
\end{observation}

%----------------------------------------------------------------------------------------------

\subsection{Boundary Detection}
\label{subsec:boundary-detection}
% !TEX root = main.tex
 %please leave the above line untouched. Thanks, Billy.

$\mathtt{Boundary-Detection}$ is a parameterized procedure run by boundary nodes with common chirality to tell each such node $b$, for each boundary $B$ that $b$ lies on, whether $B$ is an inner or outer boundary. This procedure is a modification of the first phase of the algorithm presented in Bazzi and Briones~\cite{BB18}, specifically adapting their subroutine $\mathtt{Stretch Expansion}$ to handle (1) inner boundaries and (2) an edge case that may not be needed (and is not addressed) in~\cite{BB18} but is needed here where two adjacent segments both have count $6$ but are lexicographically different (see Figure~\ref{fig:stretch-expansion-edge-case}).\footnote{Recall that~\cite{BB18} is a brief announcement and this edge case may be handled in the full version of their paper.} These adaptations result in subroutines $\mathtt{Inner-Stretch-Expansion}$ and $\mathtt{Outer-Stretch-Expansion}$, respectively. Modifications to the pseudocode from Bazzi and Briones~\cite{BB18} are highlighted with blue text.

For the sake of self-containment, the entire modified boundary detection procedure is described here (translating some of the terms used in~\cite{BB18} to the terms of the current paper.\footnote{In particular, ``stretches" are translated to segments. Moreover, in \cite{BB18}, each boundary node $b$ had virtual ``nodes" associated with each boundary $b$ was on. For a given boundary, the sum of the labels of these ``nodes" of $b$ is equivalent to the boundary count of $b$.} 
Initially, each boundary node $b$ sets $seg\_head = true$ w.r.t.\ every boundary $B$ that $b$ lies on. Each node $b$, for each boundary $B$ that $b$ lies on, maintains two segments, one participating only in $\mathtt{Outer-Stretch-Expansion}$ and the other only in $\mathtt{Inner-Stretch-Expansion}$ until $b$ receives a termination message for each boundary it lies on. For a given segment $s$ on boundary $B$, if $\mathtt{DetectTermination()}$ returns $true$, then there are $6/|s.count|$ segments on the boundary and the head of each segment will have $seg\_head = true$ w.r.t.\ $B$.  Furthermore, for boundary $B$, exactly one of $\mathtt{Outer-Stretch-Expansion}$ or $\mathtt{Inner-Stretch-Expansion}$ successfully called $\mathtt{DetectTermination()}$, corresponding to $B$ being an outer boundary or inner boundary, respectively. Subsequently, $s$ can send a message to each of its $6/|s.count|$ next segments  informing them that it is ready to terminate. Once the head of $s$ receives $6/|s.count|$ such messages from its previous segments on the boundary, it informs all nodes in $s$ to terminate the procedure.

$\mathtt{Outer-Stretch-Expansion}$, initiated by a segment $s$ attempts to merge $s$ with its next segment $s'$. The end goal of repeatedly running this subroutine is to form $k \in \{1,2,3,6\}$ segments along the outer boundary, each with count $6/k$. When $s$ has $count > 0$, it invokes the subroutine and compares its count to $s'$'s count. If $s.count> s'.count$ and $s.count + s'.count \leq 6$, then $\mathtt{Merge(s,s')}$ is called. If $s.count = s'.count$ and $s$ and $s'$ are lexicographically equal, then $\mathtt{DetectTermination()}$ is called. If $s.count = s'.count$, $s$ is lexicographically greater than $s'$, and $s.count + s'.count \leq 6$, then $\mathtt{Merge(s,s')}$ is called. $\mathtt{Inner-Stretch-Expansion}$ is similar to $\mathtt{Outer-Stretch-Expansion}$, but the conditions are modified such that the end result of repeatedly running $\mathtt{Inner-Stretch-Expansion}$ is $k \in \{1,2,3,6\}$ segments along an inner boundary, each with count $-6/k$. The pseudocode for the subroutines, with modifications from the original $\mathtt{Stretch Expansion}$ in blue, is given in Subroutine~\ref{prot:inner-stretch-expansion} and Subroutine~\ref{prot:outer-stretch-expansion}.

$\mathtt{Merge(s,s')}$, initiated by segment $s$, consists of the head of $s'$ setting $seg\_head = false$ and the head of $s$ updating its count to $s.count + s'.count$. 

When a segment $s$ runs $\mathtt{DetectTermination()}$, it lexicographically compares itself with the previous $6/|s.count|$ segments one by one. If it is lexicographically equal to each of them (and their $count$ value did not change during the subroutine), then $\mathtt{DetectTermination()}$ returns $true$.

\begin{figure}
	\includegraphics[page=18,height=2.2in]{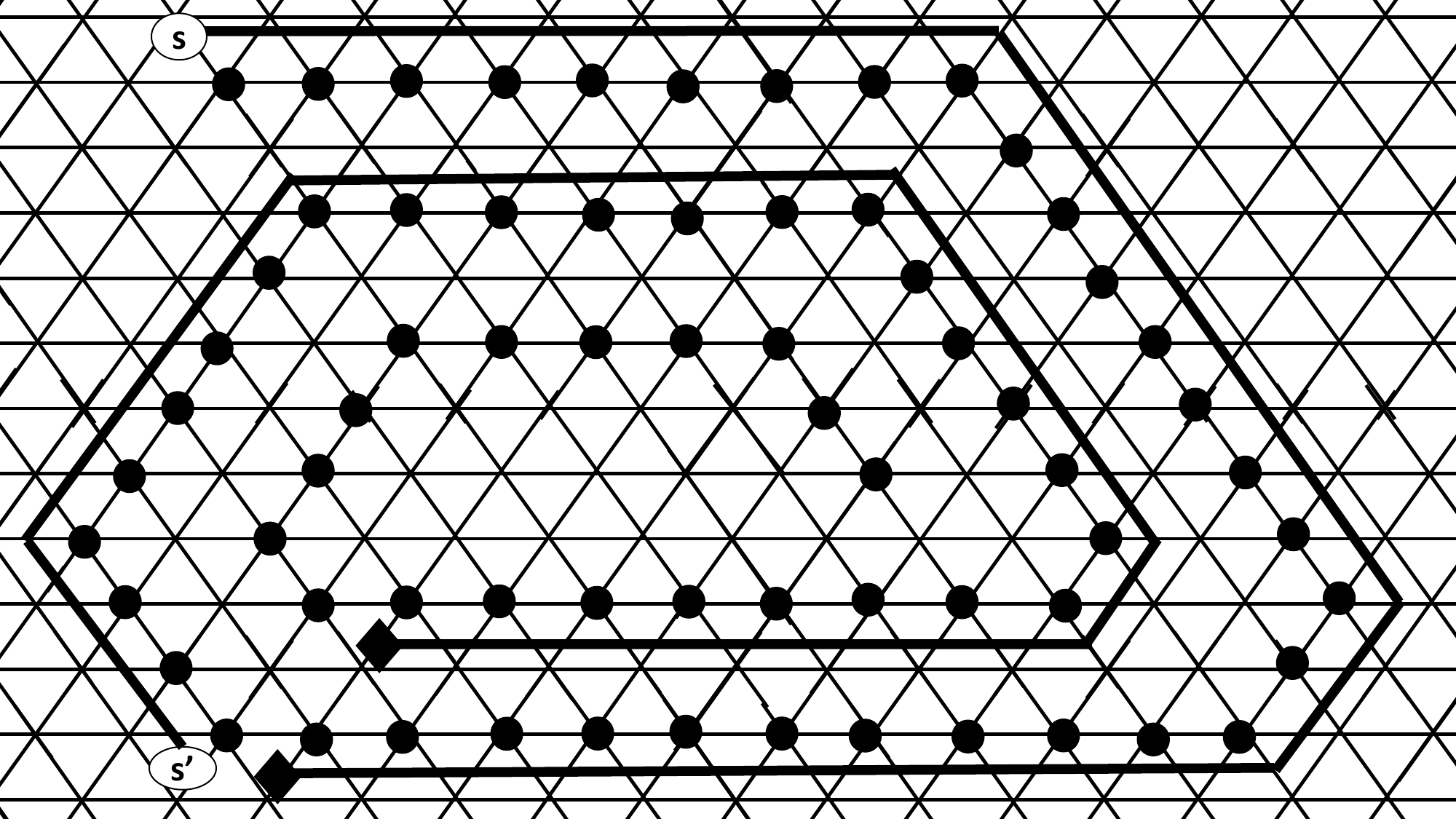}
	\caption{Edge case handled by $\mathtt{Outer-Stretch-Expansion}$. Segments $s$ and $s'$ both have count $6$ but $s$ is lexicographically larger than $s'$. They should not merge.} \label{fig:stretch-expansion-edge-case}
\end{figure} 

\alglanguage{pseudocode}
\begin{algorithm}
	\caption{Outer-Stretch-Expansion}
	\label{prot:outer-stretch-expansion}
	\begin{algorithmic}[1]	
		\Function{AttemptExpansion}{}
		\State $\triangleright$ $s$ and $s'$ are two adjacent segments where $s'$ is after $s$.
		\If {$s.count > s'.count \wedge (s.count + s'.count \leq 6 \wedge s.count >0)$}
			\State Merge($s,s'$)
		\ElsIf{$s.count = s'.count = 1, 2, 3,$ or $6$}
			\If{$s \equiv s'$}
				\State DetectTermination()
			\ElsIf{$s > s' {\color{blue}\wedge (s.count + s'.count \leq 6)}$}
				\State Merge($s,s'$)
			\EndIf
		\EndIf
		\EndFunction
		\Statex
	\end{algorithmic}
\end{algorithm}

\alglanguage{pseudocode}
\begin{algorithm}
	\caption{Inner-Stretch-Expansion}
	\label{prot:inner-stretch-expansion}
	\begin{algorithmic}[1]	
		\Function{AttemptExpansion}{}
		\State $\triangleright$ $s$ and $s'$ are two adjacent segments where $s'$ is after $s$.
		\If {{\color{blue}$s.count < s'.count \wedge (s.count + s'.count \geq -6 \wedge s.count <0)$}}
		\State Merge($s,s'$)
		\ElsIf{$s.count = s'.count =$ {$\color{blue}-1, -2, -3, \text{ or } -6$}}
		\If{$s \equiv s'$}
		\State DetectTermination()
		\ElsIf{{\color{blue}$s < s' \wedge (s.count + s'.count \geq -6)$}}
		\State Merge($s,s'$)
		\EndIf
		\EndIf
		\EndFunction
		\Statex
	\end{algorithmic}
\end{algorithm}

\begin{theorem}\label{the:boundary-detection}
When executed by contracted boundary particles, procedure $\mathtt{Boundary-Detection}$ terminates in $O(L^2_{\max})$ rounds resulting in each boundary node $b$ knowing, for each boundary $B$ it is on, whether $B$ is an inner or outer boundary. If $b$ has $seg\_head = true$ w.r.t.\ boundary $B$, then $b$ knows how many nodes $k$, $k \in \{1,2,3,6\}$, are also segment heads w.r.t.\ $B$.
\end{theorem}

\begin{proof}
We first prove that every node on a boundary $B$ identifies $B$ correctly as either an outer or inner boundary and that the procedure terminates. The case of the outer boundary is analyzed and it is noted that the analysis for an inner boundary is similar.

 It is easy to see that any segment with count $0$ or $-1$ will not initiate $\mathtt{Merge()}$. Note that deadlocks, i.e. a situation where each segment tries to execute $\mathtt{Merge()}$ with its next segment, are avoided when at least one segment has count $0$ or $-1$. Thus, eventually, only segments with positive counts will remain. It is easy to see that the sum of the boundary counts of all outer boundary nodes is $6$. Thus, there will be at most $6$ segments finally covering the boundary nodes, each with count $\geq 1$. Due to the checks in place in $\mathtt{Outer-Stretch-Expansion}$, the count of a segment never exceeds $6$. 

If the boundary is covered by several positive count segments which are not lexicographically equal, then $\mathtt{DetectTermination()}$ will not return $true$. Instead, the segments will continue to run $\mathtt{Outer-Stretch-Expansion}$. Deadlock w.r.t. $\mathtt{Merge()}$ operations is avoided in this case because there will always exist a segment that is lexicographically less than both its previous and next segments. 

$\mathtt{DetectTermination()}$ run by a segment $s$ with count  $i$ returns $true$ when the next $6/|i|$ segments after $s$ are lexicographically equal to $s$. Thus, the number of segments that finally cover the boundary is a divisor of $6$, i.e. $1,2,3,$ or $6$. Furthermore, each head of a segment with count $i$ has $seg\_head = true$ w.r.t.\ the boundary and knows the number of segments as $6/|i|$.

One instance of $\mathtt{Outer-Stretch-Expansion}$, run on two segments $s$ and $s'$ of lengths $\ell$ and $\ell'$ respectively, takes $O(\min \lbrace\ell,\ell'\rbrace)$ rounds. The length of any segment is upper bounded by $L_{\max}$. Thus the running time of $\mathtt{Outer-Stretch-Expansion}$ is at most $O(L_{\max})$ rounds. Similarly for $\mathtt{Inner-Stretch-Expansion}$, its max running time is $O(L_{\max})$ rounds. One instance of lexicographic comparison or one instance of $\mathtt{Merge(s,s')}$ of two adjacent segments also takes $O(L_{\max})$ rounds. $\mathtt{DetectTermination()}$ compares one segment with at most 6 other segments lexicographically and thus takes $O(L_{\max})$ rounds as well.

Now, the running time of $\mathtt{Boundary-Detection}$ is upper bounded by considering the time taken to merge all nodes on the boundary into one segment and subsequently run $\mathtt{DetectTermination()}$ and subsequently terminate execution. Either $\mathtt{Outer-Stretch-Expansion}$ or $\mathtt{Inner-Stretch-Expansion}$ is called $O(L_{\max})$ times, each instance taking $O(L_{\max})$ rounds resulting in a running time of $O(L^2_{\max})$ rounds.
\end{proof}

%----------------------------------------------------------------------------------------------

\subsection{Leader Election on a Convex Polygon without Sharp Vertices}
\label{subsec:LE-convex-polygon}
% !TEX root = main.tex
 %please leave the above line untouched. Thanks, Billy.
 
Procedure $\mathtt{Convex-Polygon-Leader-Election}$ relies on three subroutines, which are described below. Note that the outer boundary nodes of a convex polygon without sharp vertices form a hexagon in the grid. Let $b$ be a vertex, occupied by particle $P$, with successor node $d$ w.r.t.\ the outer boundary. Define \emph{$P$'s side} as the side of the hexagon containing $b$ and $d$. 

Let LSLS stand for largest same length sides and SSLS stand for smallest same length sides. Every possible hexagon is isomorphic to one of the following four. See Figures~\ref{fig:cat1-hexagon-1ssls}-\ref{fig:cat4-hexagon}.
\begin{enumerate}
\item Category 1: The hexagon has exactly either 1 LSLS or 1 SSLS.
\item Category 2: The hexagon has either 2 LSLS and 4 SSLS, 4 LSLS and 2 SSLS, or 2 LSLS, 2 SSLS, and 2 other same length sides.
\item Category 3: The hexagon has exactly 3 LSLS and 3 SSLS.
\item Category 4: The hexagon has 6 sides of the same length.
\end{enumerate}

\begin{figure}
	\includegraphics[page=6,height=2.2in]{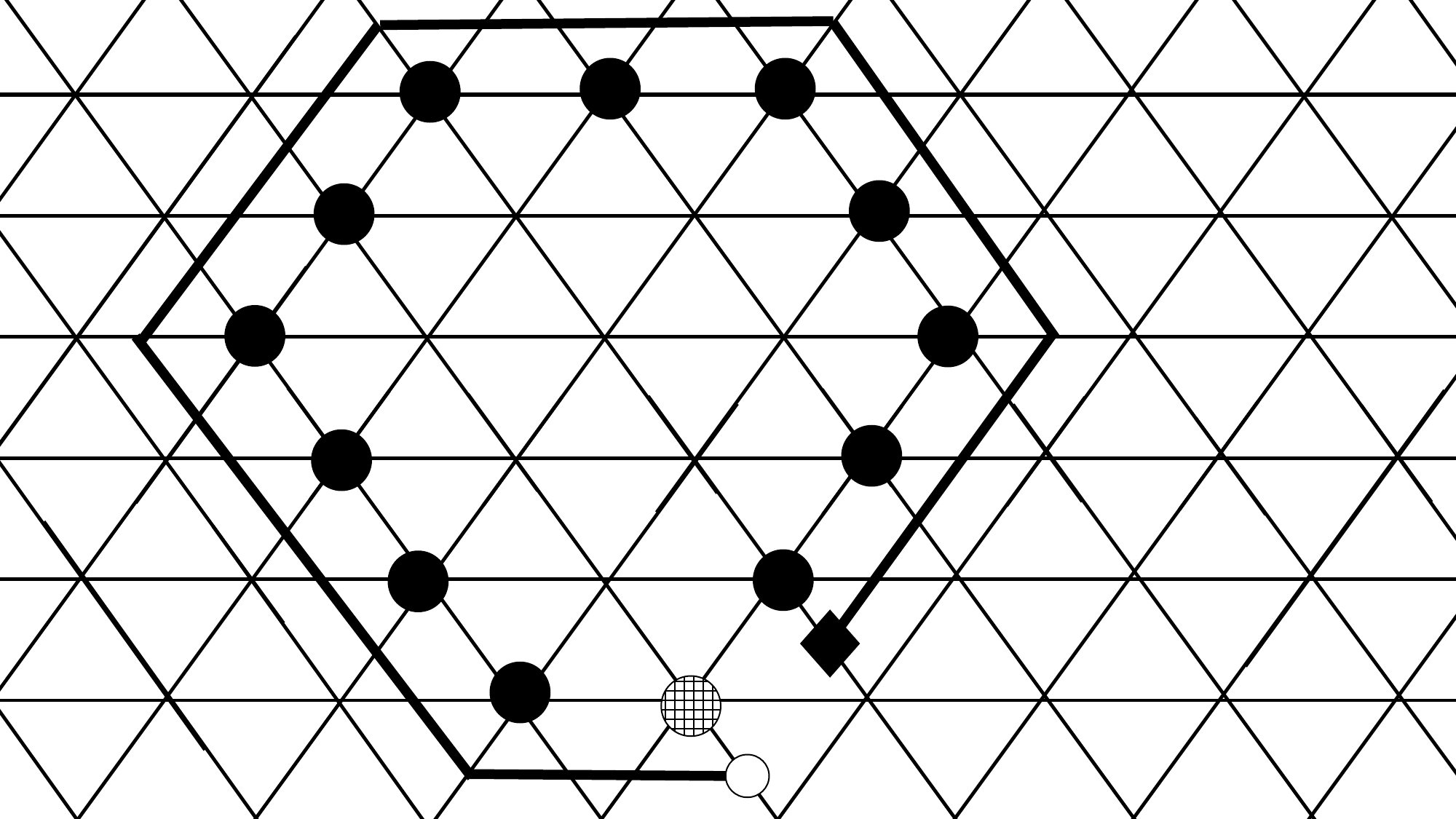}
	\caption{Category 1 hexagon of CW particles with 1 SSLS. Exactly one seg-head, denoted by patterned particle. Non-seg-head particles are filled. Empty circle denotes head of segment. Filled diamond denotes end of segment.} \label{fig:cat1-hexagon-1ssls}
\end{figure} 

\begin{figure}
	\includegraphics[page=7,height=2.2in]{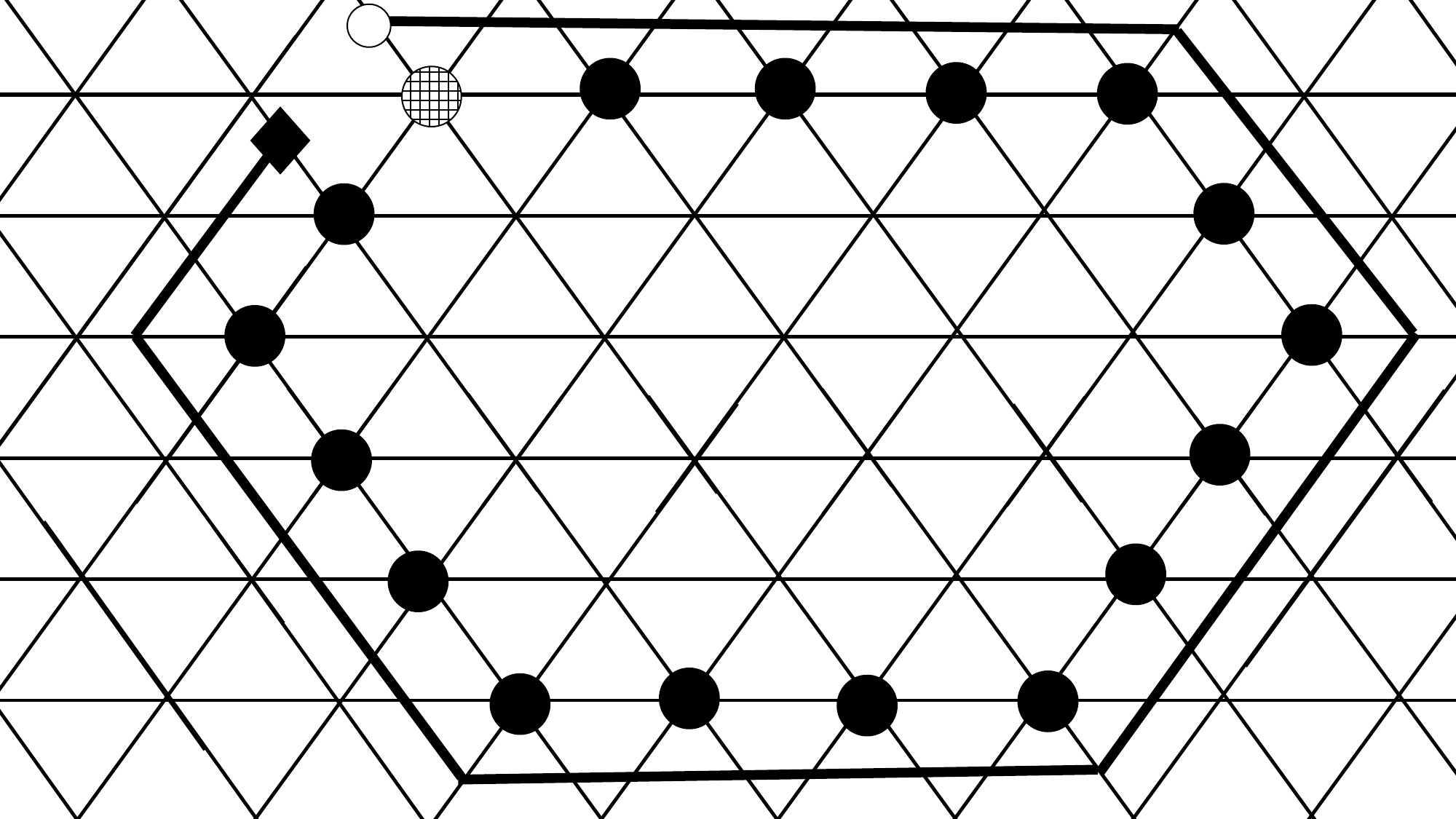}
	\caption{Category 1 hexagon of CW particles with 1 LSLS and exactly one seg-head.} \label{fig:cat1-hexagon-1lsls}
\end{figure}  

\begin{figure}
\begin{subfigure}{.5\textwidth}
  \centering
  \includegraphics[page=9,width=.8\linewidth]{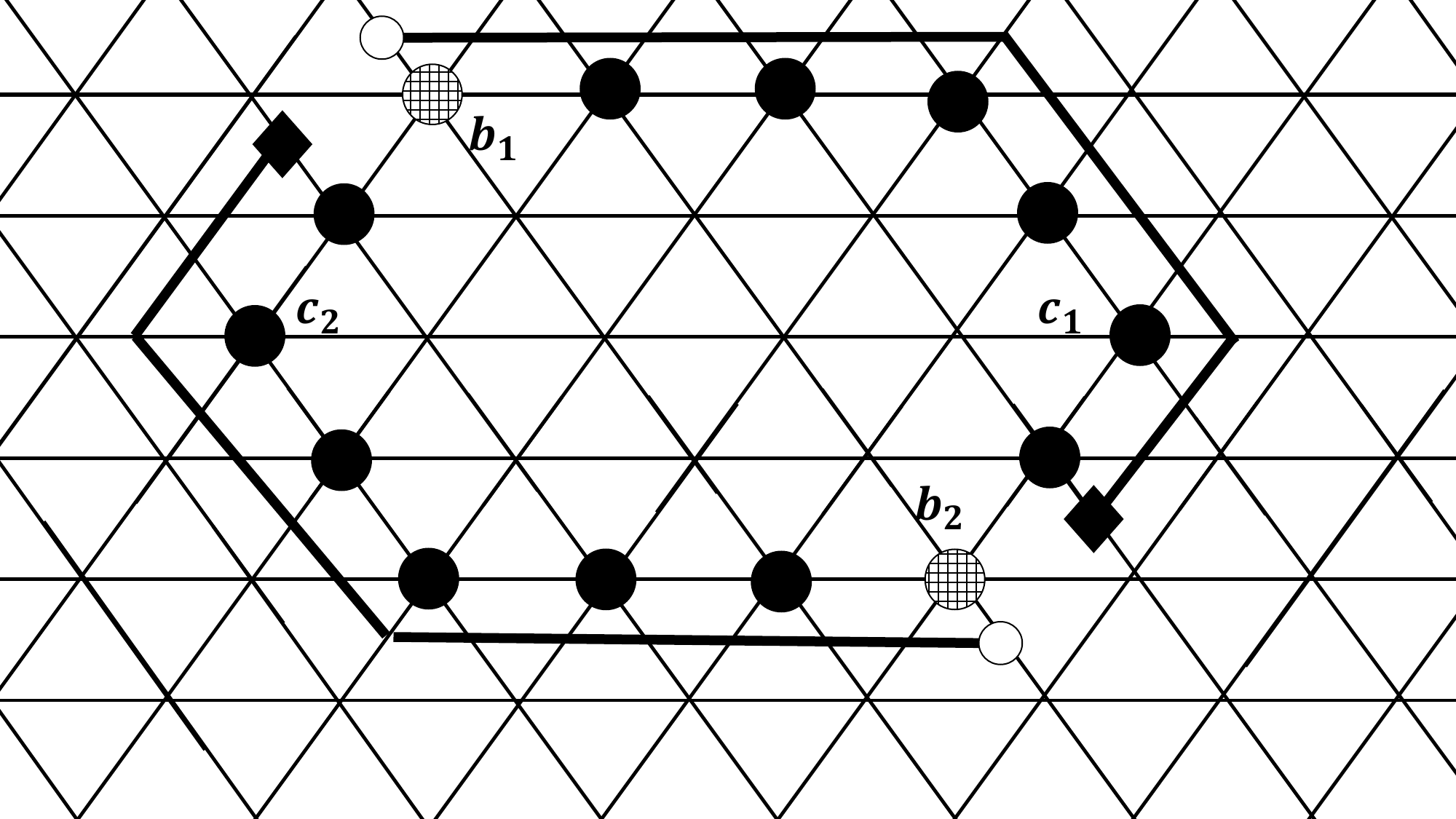}
  \caption{Initially.}
  \label{fig:cat2-hexagon-2lsls-4ssls-before}
\end{subfigure}%
\begin{subfigure}{.5\textwidth}
  \centering
  \includegraphics[page=10,width=.8\linewidth]{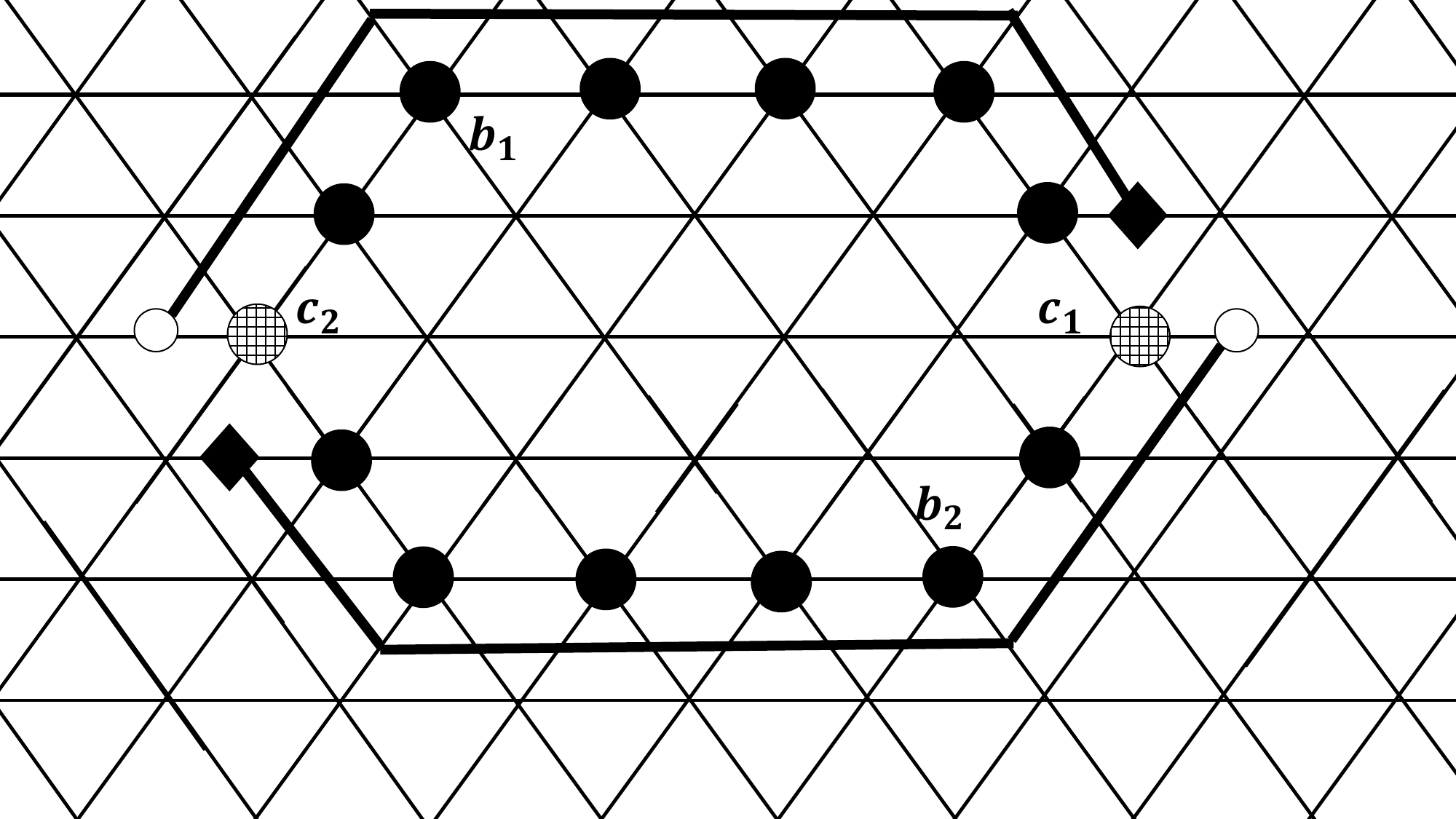}
  \caption{After running $\mathtt{Mid-Line}$.}
  \label{fig:cat2-hexagon-2lsls-4ssls-aftermidline}
\end{subfigure}
\caption{Category 2 hexagon of CW particles with 2 LSLS and 4 SSLS and two seg-heads. Nodes $b_1$, $b_2$, $c_1$, and $c_2$ are occupied by particles $P_1$, $P_2$, $Q_1$, and $Q_2$ respectively and are marked.}
\label{fig:cat2-hexagon-2lsls-4ssls}
\end{figure}

\begin{figure}
\begin{subfigure}{.5\textwidth}
  \centering
  \includegraphics[page=11,width=.8\linewidth]{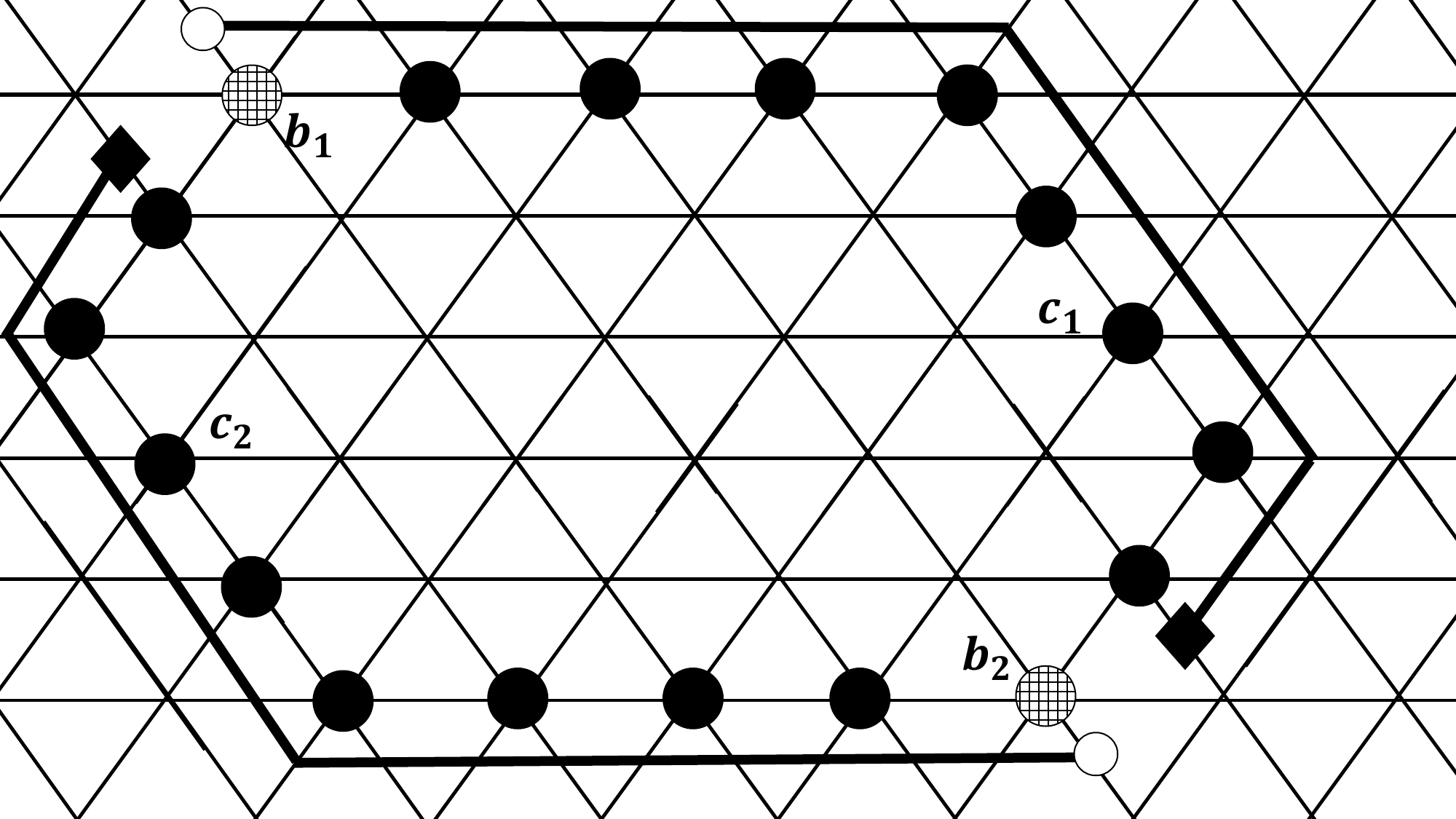}
  \caption{Initially.}
  \label{fig:cat2-hexagon-2lsls-2ssls-before}
\end{subfigure}%
\begin{subfigure}{.5\textwidth}
  \centering
  \includegraphics[page=12,width=.8\linewidth]{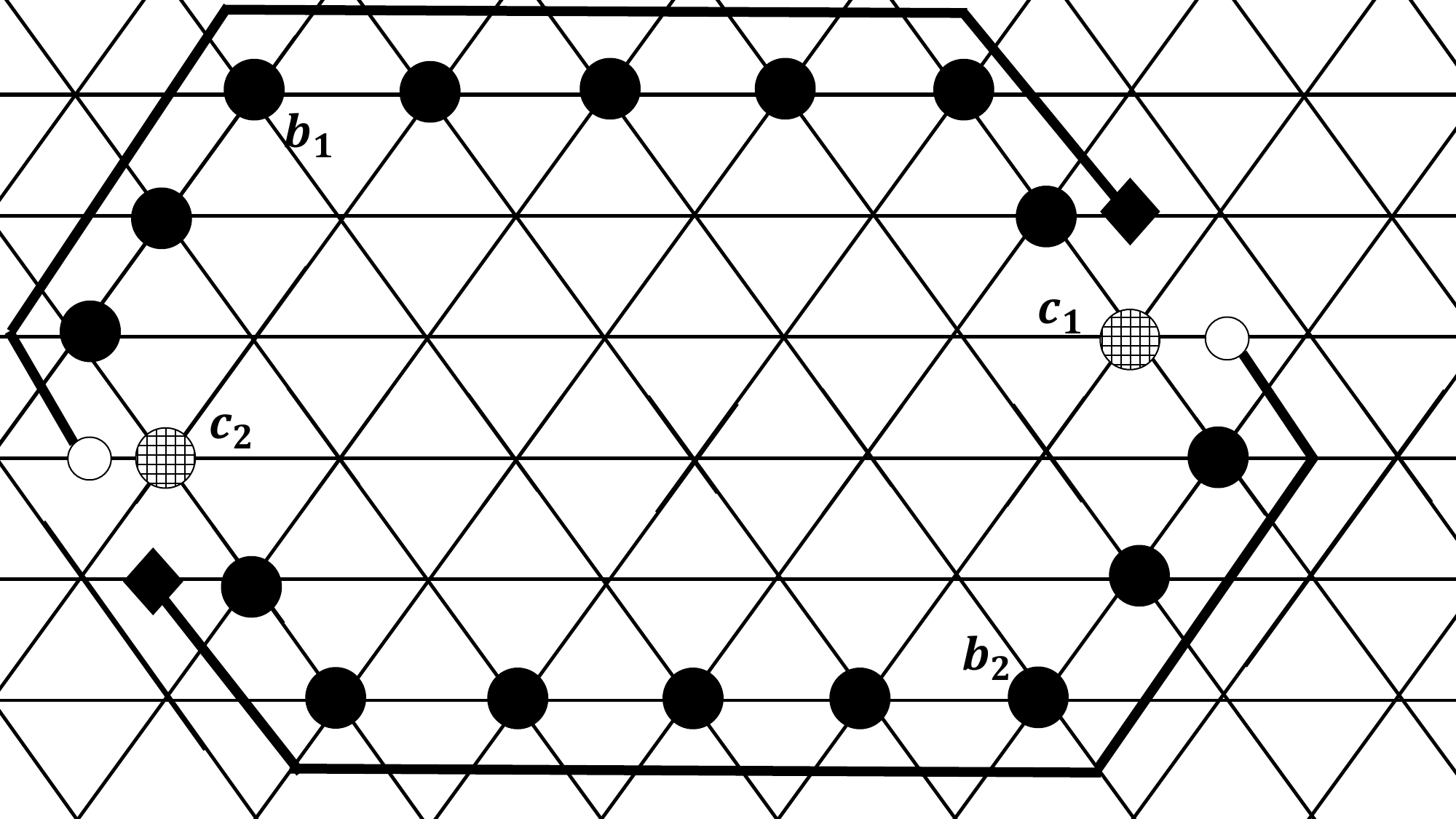}
  \caption{After running $\mathtt{Mid-Line}$.}
  \label{fig:cat2-hexagon-2lsls-2ssls-aftermidline}
\end{subfigure}
\caption{Category 2 hexagon of CW particles with 2 LSLS, 2 SSLS, and 2 other same length sides and two seg-heads. Nodes $b_1$, $b_2$, $c_1$, and $c_2$ are occupied by particles $P_1$, $P_2$, $Q_1$, and $Q_2$ respectively and are marked.}
\label{fig:cat2-hexagon-2lsls-2ssls}
\end{figure}

\begin{figure}
	\includegraphics[page=13,height=2.2in]{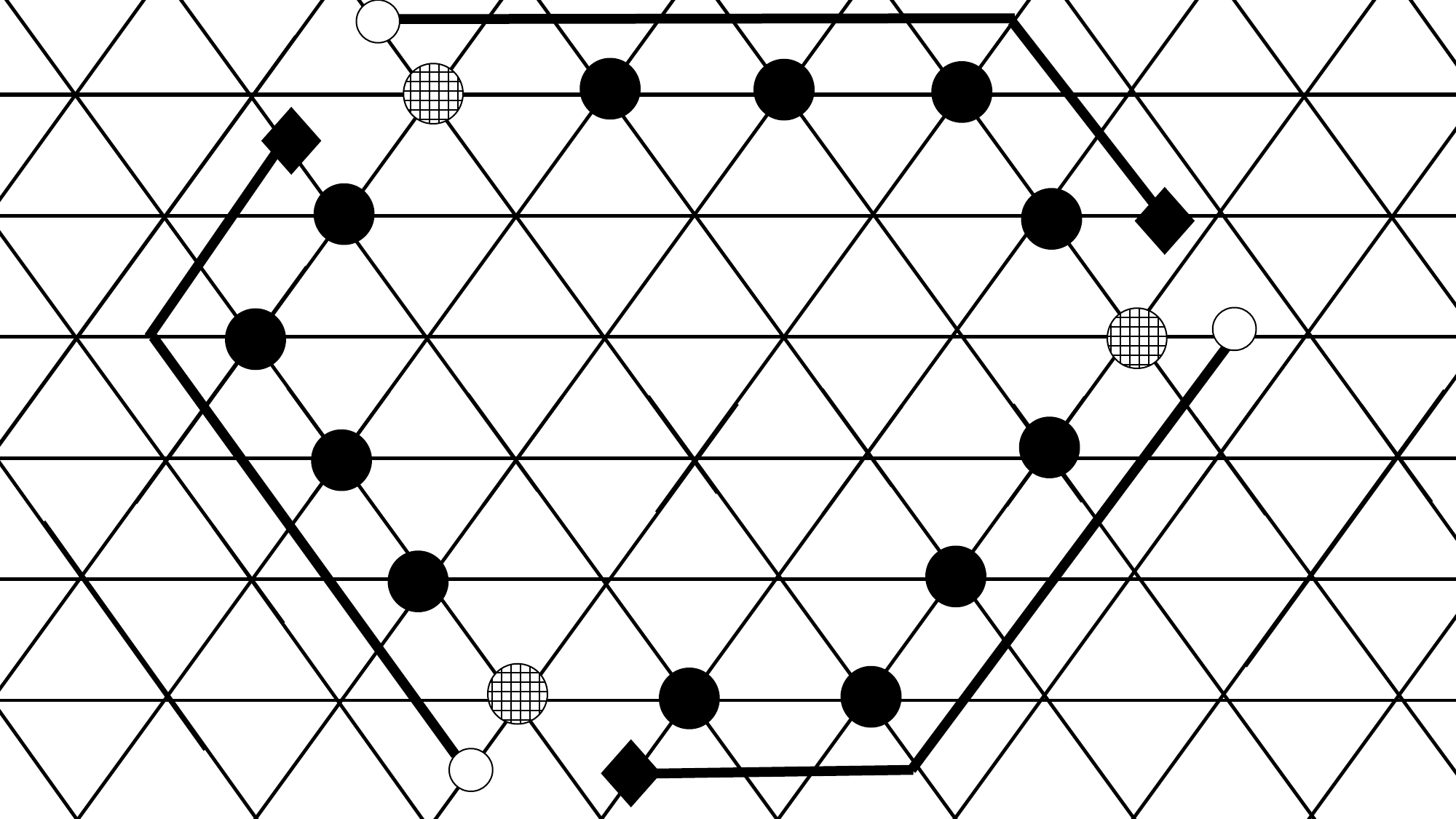}
	\caption{Category 3 hexagon of CW particles with three seg-heads.} \label{fig:cat3-hexagon}
\end{figure} 

\begin{figure}
	\includegraphics[page=8,height=2.2in]{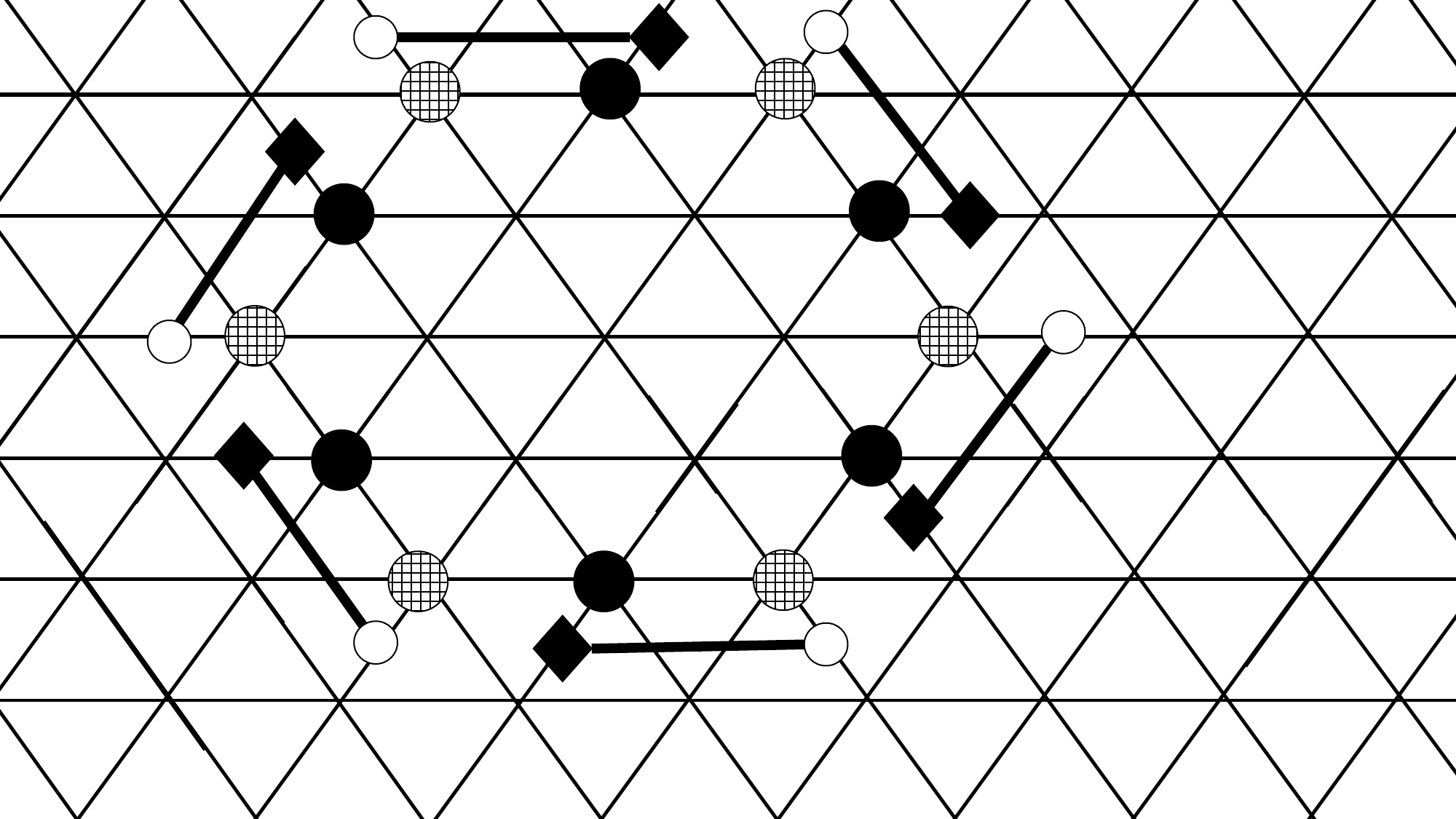}
	\caption{Category 4 hexagon of CW particles with six seg-heads.} \label{fig:cat4-hexagon}
\end{figure} 

Let $P_0, P_1, \ldots P_{k-1}$ be the sequence of particles with $seg\_head = true$ on the outer boundary such that $P_{(i+1) \mod k}$ is the next seg-head after $P_i$. Subroutine $\mathtt{Compare-Length(x)}$ is initiated by a seg-head $P_i$ to compare the length of $P_i$'s segment with that of $P_{(i+x) \mod k}$'s segment.\footnote{This is similar to the lexicographic comparison of two segments in~\cite{BB18}. However, in lexicographic comparison, unlike in $\mathtt{Compare-Length(x)}$, the boundary count of each node is also used for comparison.}  The procedure simulates the way a Turing machine would perform a similar task, where the segments would be segments of the machine's tape (refer to \cite{DFSVY18} for an example). 
Note that $P_i$ specifies messages for $P_{(i+x) \mod k}$ by encoding $x$ into the message. Each $P_j$, $i \leq j \leq (i+x) \mod k$ can increment a counter, also embedded in the message, until the destination particle is reached. When $k$ is a constant, the size of these encodings is only a constant number of bits. The following lemma captures the running time of the subroutine.

\begin{lemma}\label{lem:running-time-compare-length}
	When $x$ is a constant and there are $L$ nodes on the outer boundary, if the nodes from $P_i$'s segment's head to $P_{(i+x) \mod k}$'s segment's tail run $\mathtt{Compare-Length(x)}$, then the subroutine terminates in $O(L^2)$ rounds, resulting in $P_i$ knowing the size comparison between the two segments.
\end{lemma}

\begin{proof}
	The correctness is trivial. Regarding time complexity, a message from $P_i$ to $P_{(i+x) \mod k}$ or vice-versa travels through at most $O(L)$ nodes since $x$ is a constant, thus taking $O(L)$ rounds. The length of either segment is at most $L$, so at most $L$ comparisons must be made. Thus it takes $O(L^2)$ rounds for $\mathtt{Compare-Length(x)}$ to terminate.
\end{proof}

Consider two parallel lines $M_1$ and $M_2$ of the grid. The \emph{mid-line(s)} between $M_1$ and $M_2$ is the line(s) parallel to both $M_1$ and $M_2$ which is either equidistant from both $M_1$ and $M_2$ or not closer to one of the lines by more than a unit distance. Consider a category 2 hexagon where opposite outer boundary vertices $b_1$ and $b_2$, occupied by particles $P_1$ and $P_2$ respectively, have $seg\_head = true$ and the remaining nodes have $seg\_head = false$. There exist either 1 or 2 mid-lines between $P_1$'s side and $P_2$'s side. Let $c_1$ and $c_2$, occupied by particles $Q_1$ and $Q_2$ respectively, be nodes on $P_1$ and $P_2$'s segments respectively lying on the mid-line (or on the closer mid-line in the case of 2 mid-lines). The outer boundary particles run subroutine $\mathtt{Mid-Line}$ to find $c_1$ and $c_2$ and subsequently $Q_1$ and $Q_2$ set $seg\_head = true$ and $P_1$ and $P_2$ set $seg\_head = false$. See Figures~\ref{fig:cat2-hexagon-2lsls-4ssls} and~\ref{fig:cat2-hexagon-2lsls-2ssls} for examples.

The subroutine works as follows. Consider $P_1$'s segment (the process is similar for $P_2$'s segment). Particle $P_1$ sends a message along its segment telling all nodes from the next vertex in its segment, $d_1$, to the tail of the segment, $d_2$, to mark themselves. Once $P_1$ receives an acknowledgement that this is done, it sends a message to $d_1$ instructing it do the following. The node $d_1$ unmarks itself and sends a message to $d_2$. Node $d_2$ then unmarks itself and sends a message through its predecessors to the furthest marked node. The previous process is repeated until exactly one node $c_1$ is left marked. This node sets $seg\_head = true$ and sends a message to $b_1$ to set $b_1$'s $seg\_head$ to $false$. Finally, $c_1$ sends a termination message to $c_2$. Once $c_1$ receives a termination message from $c_2$, $Q_1$ informs all nodes in its segment to terminate execution of the subroutine. The following observation captures the running time of $\mathtt{Mid-Line}$.

\begin{observation}\label{obs:running-time-mid-line}
	Let there be $L$ outer boundary nodes on a category 2 hexagon with opposite vertices $b_1$ and $b_2$, occupied by particles $P_1$ and $P_2$ respectively, with $seg\_head = true$ and remaining nodes with $seg\_head = false$. Subroutine $\mathtt{Mid-Line}$, run by the $L$ nodes, terminates in $O(L^2)$ rounds, such that nodes $c_1$ and $c_2$, which are the closest nodes in $P_1$ and $P_2$'s segments lying on mid-lines between $P_1$'s side and $P_2$'s side respectively, now have $seg\_head = true$ and $b_1$ and $b_2$ have $seg\_head = false$.
\end{observation}

Intuitively, when the segment heads are on the mid-line as promised by Observation 
 \ref{obs:running-time-mid-line}, if they move towards the center, they can get next to each other and elect one of them as a leader. The final subroutine, $\mathtt{Snake-Movement(D,x)}$, is used for election in hexagons of several types.
Consider a path $p$ of $w$ nodes occupied by contracted particles with head node $b$ occupied by particle $P$ and tail node $c$. $P$ has $seg\_head = true$ w.r.t.\ the outer boundary and the remaining particles have $seg\_head = false$. See Figure~\ref{fig:snake-example} for an example of a path. The subroutine is run by the particles in $p$ to expand path $p$ (termed \emph{snake $p$}) $x$ nodes in direction $D$ through $b$ without breaking $p$'s connectivity and while keeping the tail node of $p$ fixed at $c$. Here $x$ is restricted to $x \leq w$. The particles of $p$ move in a series of expansions and contractions in the direction $D$ through $b$ while using the particles of $p$ like a Turing machine to maintain a counter $y$ of how many nodes remain to be moved through. 

\begin{figure}
	\includegraphics[page=34,height=2.2in]{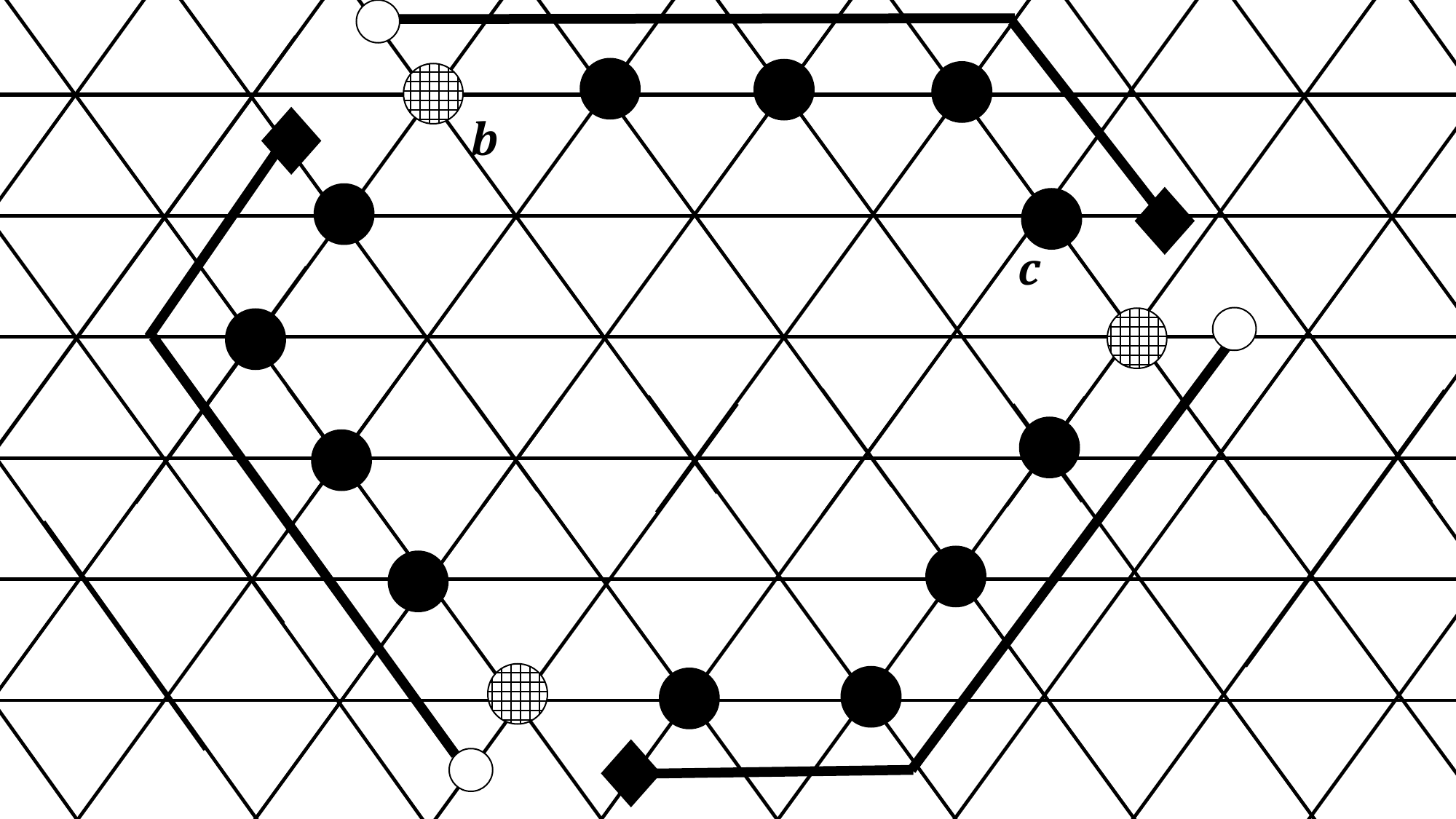}
	\caption{Example of a path $p$ of $5$ nodes with head $b$ and tail $c$. Particle $P$ occupies $b$ and has $seg\_head = true$.} \label{fig:snake-example}
\end{figure} 

$P$ sends a message throughout $p$ indicating each node $b$ should maintain parent and child pointers to $b$'s predecessor and successor in $p$, respectively.\footnote{Note that if $P$ occupies two nodes, $P$'s parent is the predecessor of its head and $P$'s child is the successor of its tail.} $P$'s parent pointer and $c$'s child pointer are both set to $\bot$. Furthermore, $P$ sets flag $head\_of\_snake = true$. $P$ initializes $y$ to $x$ and repeats the following until $y=0$. If $P$ is contracted, it expands in the direction $D$ and decrements $y$. If $P$ is expanded, $P$ checks if it can pull the particle $Q$ occupying its child as follows. If $Q$ is contracted, $P$ pulls $Q$. Else $P$ sends a message to $Q$ to pull the particle at $Q$'s child. This process is recursively performed in $p$ until some particle in $p$ pulls the particle at its child. 

Furthermore, if $P$ wishes to expand in a given direction and the target node is already occupied by another particle $R$, then $P$ decrements $y$, passes the values of $D$ and $y$ to $R$, and then sets $head\_of\_snake = false$ and $P$'s parent pointer to the node occupied by $R$. $R$ sets its parent pointer to $\bot$, sets its child pointer to the node occupied by $P$, and sets $head\_of\_snake = true$. Subsequent movements of $p$ are led by $R$. The running time of the procedure is captured by the following lemma and a proof sketch is given as the proof is straightforward.

Note that only a segment on the outer boundary can perform procedure $\mathtt{Snake-Movement(D,x)}$. Hence, the instruction to move is not contradicted by an instruction of another segment that snake $p$ particles may participate in.

\begin{lemma}\label{lem:running-time-snake-movement}
Assume that $L$ contracted particles of a snake run $\mathtt{Snake-Movement(D,x)}$, where $x \leq L$. Then, the subroutine terminates in $O(x^2)$ rounds without breaking connectivity. On termination, either the snake head reached a node at distance $x$ away from the head of the snake in direction $D$, or the next particle in direction $D$ belongs to another snake.
\end{lemma}

\begin{proof}[Proof Sketch]
Consider the sequence of particles in the snake, starting from the head of the snake, consisting of only expanded particles until the first contracted particle is found. Let this sequence consist of $k$ expanded particles $P_1$ to $P_k$. In order for the particle occupying the head of the snake to expand to the next node, a series of contractions must first occur, starting with $P_k$ contracting and pulling. Subsequently $P_{k-1}$ contracts and pulls $P_k$ and so on until at last $P_1$ contracts. Then $P_1$ is free to expand. This sequence of pulls takes at most $k$ rounds and one additional round for $P_1$ to expand. After $P_1$ expands, it takes $O(x)$ rounds to decrement the counter $y$ by one. Initially, it took $O(x)$ rounds to check if $y=0$ and $k$ rounds for the head to transmit a message that it wants to expand to $P_k$, thus starting this chain of contractions. 

After such a sequence of pulls and an expansion, the length of the sequence of expanded particles increases from $k$ to $k+1$. Define the series of rounds taken to increase the length of this sequence by 1 as one stage. Now, there are at most $x$ stages, for a total running time of $\sum_{i=0}^{x} O(i + x) = O(x^2)$ rounds. Note that during $\mathtt{Snake-Movement(D,x)}$, it is possible for the head to want to expand into an already occupied node. This would only reduce the running time, as this means that we do not need to expand the original head and possibly save several rounds of expansions and contractions within the snake.

We are careful when calling the procedure to ensure that no particle is shared between 2 or more snakes. Thus the particles in a snake will not perform contrasting movements. Therefore, the correctness of the lemma is trivial.
\end{proof}

Now the procedure $\mathtt{Convex-Polygon-Leader-Election}$ is described. Initially, the $6$ particles that occupy vertices on the outer boundary set $seg\_head = true$ while the remaining particles in the polygon set $seg\_head = false$. Each of these $6$ particles initiates $\mathtt{Compare-Length(x)}$ for $1 \leq x \leq 6$, sends messages to the remaining $5$ particles with the results of these comparisons, and determines which category hexagon it lies on.\footnote{With this information, a particle can compute, using a constant amount of space, the total order on the lengths of sides of the hexagon. Combined with the information of which sides are equal in length, a particle can determine both the category of the hexagon it lies on and the type of its side.} The procedure follows one of the following four cases:
\begin{enumerate}
\item \textit{Category 1 hexagon:} Let $P$ occupy a vertex such that $P$'s side is the smallest or largest side depending on the hexagon. The remaining $5$ vertices set $seg\_head = false$ and $P$ becomes the leader, as seen in Figure~\ref{fig:cat1-hexagon-1ssls} and Figure~\ref{fig:cat1-hexagon-1lsls}.

\item \textit{Category 2 hexagon:} If there are exactly 2 LSLS, then those sides' vertices  keep $seg\_head = true$ and the remaining vertices set $seg\_head = false$. Else there are 2 SSLS whose vertices keep $seg\_head = true$ while others set $seg\_head = false$. Call particles occupying vertices with $seg\_head = true$, $P_1$ and $P_2$, and denote the direction from the successor of $P_1$ to $P_1$ as $D_1$ (similarly denote $D_2$). Now, $P_1$ and $P_2$ initiate $\mathtt{Mid-Line}$ resulting in two new particles $Q_1$ and $Q_2$ setting $seg\_head = true$ and $P_1$ and $P_2$ setting $seg\_head = false$ (two examples of this process are shown in Figure~\ref{fig:cat2-hexagon-2lsls-4ssls} and Figure~\ref{fig:cat2-hexagon-2lsls-2ssls}). 

The resulting segments of $Q_1$ and $Q_2$ form snakes $p_1$ and $p_2$ with lengths $w_1$ and $w_2$ respectively that run $\mathtt{Snake-Movement(D_1,w_1)}$ and $\mathtt{Snake-Movement(D_2,w_2)}$ in directions $D_1$ and $D_2$, respectively. In addition to the usual termination conditions when running  $\mathtt{Snake-Movement(D_1,w_1)}$ and $\mathtt{Snake-Movement(D_2,w_2)}$, the subroutines also terminate when the head of $p_1$ is adjacent to that of $p_2$. Then, the two heads run $\mathtt{MIS-Selection}$ and the particle that joins the MIS becomes the leader.

\item \textit{Category 3 hexagon:} Let $P_1, P_2,$ and $P_3$ occupy vertices $b_1, b_2,$ and $b_3$ such that $P_1, P_2,$ and $P_3$'s sides are the 3 largest sides. The remaining vertices set $seg\_head = false$ (example seen in Figure~\ref{fig:cat3-hexagon}). Let $D_1, D_2,$ and $D_3$ be the directions along the angle bisectors of $b_1, b_2,$ and $b_3$ respectively toward the center of the hexagon. The two phase procedure followed by $P_1$'s segment is now described, with $P_2$'s and $P_3$'s segments following similar procedures.  

In phase one, $P_1$ coordinates the simulation of a Turing machine on its segment. Notice that $P_1$'s segment encompasses 1 SSLS and 1 LSLS with lengths $x$ and $y$ respectively. In the simulation, the values of $f = \lfloor (y - x)/3 \rfloor$,  $g = x + f$, and $q = (y-x) \mod 3$ are computed and stored in $P_1$'s segment. If $q = 2$, it increments $g$ by 1. Now $P_1$ sends a message to the particle $Q_1$ located $f$ nodes from the head of the segment, telling $Q_1$ to set $seg\_head = true$ and store $D_1$, $f$, $g$, and $q$. $P_1$ subsequently sets $seg\_head = false$. $Q_1$ is now the head of a segment. $P_2$ and $P_3$'s segments perform similar procedures resulting in particles $Q_2$ and $Q_3$ respectively becoming heads of segments. Now $Q_1$ sends a message along the outer boundary to $Q_2$ and $Q_3$ indicating that the first phase is over. Once $Q_1$ receives a similar message from both of them, the second phase begins.

In phase two, $Q_1$'s segment acts as a snake and runs $\mathtt{Snake-Movement(D_1,g)}$. If $q = 0$, all three snakes move towards the same final node. Let the particle that occupies this node first be $R$. $R$ waits until the remaining two snakes reach it and then becomes the leader. If $q \neq 0$, the final nodes occupied by the three snakes form a triangle. Let $R$ be a particle that occupies one of the triangle's nodes. $R$ waits until the other two triangle's nodes are occupied and then runs $\mathtt{MIS-Selection}$. The particle chosen to be in the MIS becomes the leader.

\item \textit{Category 4 hexagon:} All vertices have $seg\_head = true$ (e.g., Figure~\ref{fig:cat4-hexagon}). Let $P_1$ to $P_6$ be the particles occupying the vertices $b_1$ to $b_6$ of the hexagon, with sides of same length $x$. Let $D_1$ to $D_6$ be the directions along the angle bisectors of $b_1$ to $b_6$ respectively toward the center of the hexagon. For each $i \in [1,6]$, $P_i$'s segment acts as a snake and runs $\mathtt{Snake-Movement(D_i,x)}$. All snakes move toward the same node. Let the particle that occupies this node first be $R$. $R$ waits until the 5 other snakes reach it, then becomes the leader.
\end{enumerate}

\begin{theorem}\label{the:convex-poly-le}
Procedure $\mathtt{Convex-Polygon-Leader-Election}$ run by contracted particles of a convex polygon without sharp edges results in exactly one leader being elected deterministically in $O(L^2)$ rounds, where $L$ is the number of particles on the outer boundary.
\end{theorem}

\begin{proof}
Regarding correctness, we show that for each category of hexagon, a unique leader is elected. The reader can convince themselves that all types of hexagons have been accounted for in the four hexagon categories. 

It is clear that there exists a unique shortest or longest side in a category 1 hexagon and so a unique leader is chosen. 

In a category 2 hexagon, Observation~\ref{obs:running-time-mid-line} guarantees that particles are chosen such that they lie on the same mid-line or adjacent mid-lines. The distance needed to be traveled by each segment until both heads are adjacent is $\leq L/2$. Since the segments equally divide the nodes of the outer boundary, each snake has enough contracted particles such that it is possible to traverse this distance by expanding every particle in the snake. $\mathtt{MIS-Selection}$ is guaranteed to choose exactly one leader due to Observation~\ref{obs:mis-selection}.

For a category 3 hexagon, let us inscribe the hexagon in an equilateral triangle with vertices $A$, $B$, and $D$, centroid $C$, and $F$ trisecting $\overline{AB}$, as seen in Figure~\ref{fig:cat3-hexagon-relation-triangle}. Observe that each side of the triangle is of length $2x+y$ and $|\overline{AF}| = |\overline{FC}|$. When $(y-x) \mod 3 = 0$, $\overline{FC}$ coincides with a grid line and all snakes move toward $C$ using $\texttt{Snake-Movement()}$. However, if $(y-x) \mod 3 \neq 0$, the snakes move to nodes that form a triangle around the centroid, in which case $\mathtt{MIS-Selection}$ is run and Observation~\ref{obs:mis-selection} guarantees a leader is selected.

\begin{figure}
\begin{subfigure}{.5\textwidth}
  \centering
  \includegraphics[page=27,width=.8\linewidth]{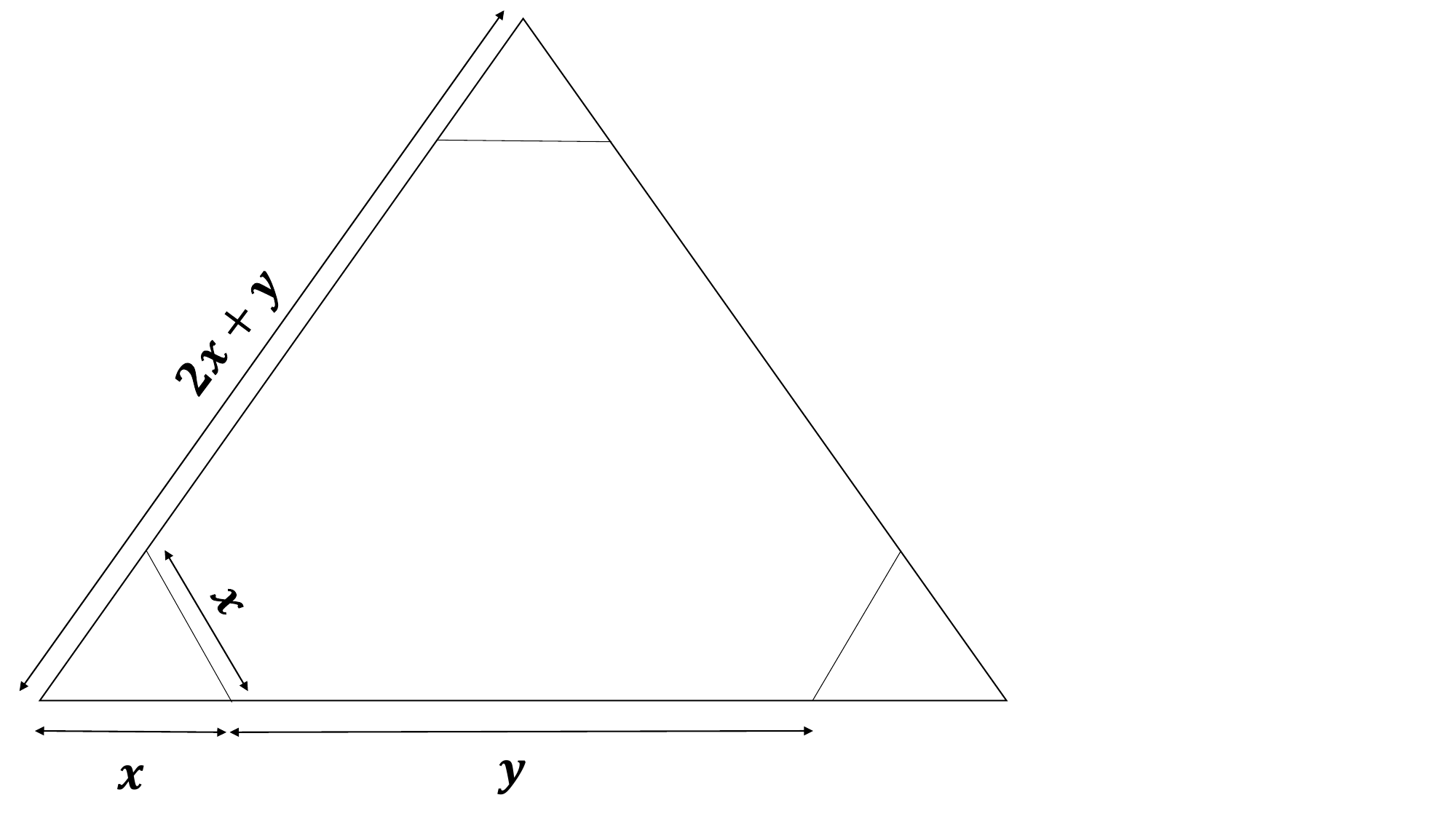}
  \caption{Hexagon inscribed in an equilateral triangle.}
  \label{fig:cat3-hexagon-inscribe-eq-triangle}
\end{subfigure}%
\begin{subfigure}{.5\textwidth}
  \centering
  \includegraphics[page=28,width=.8\linewidth]{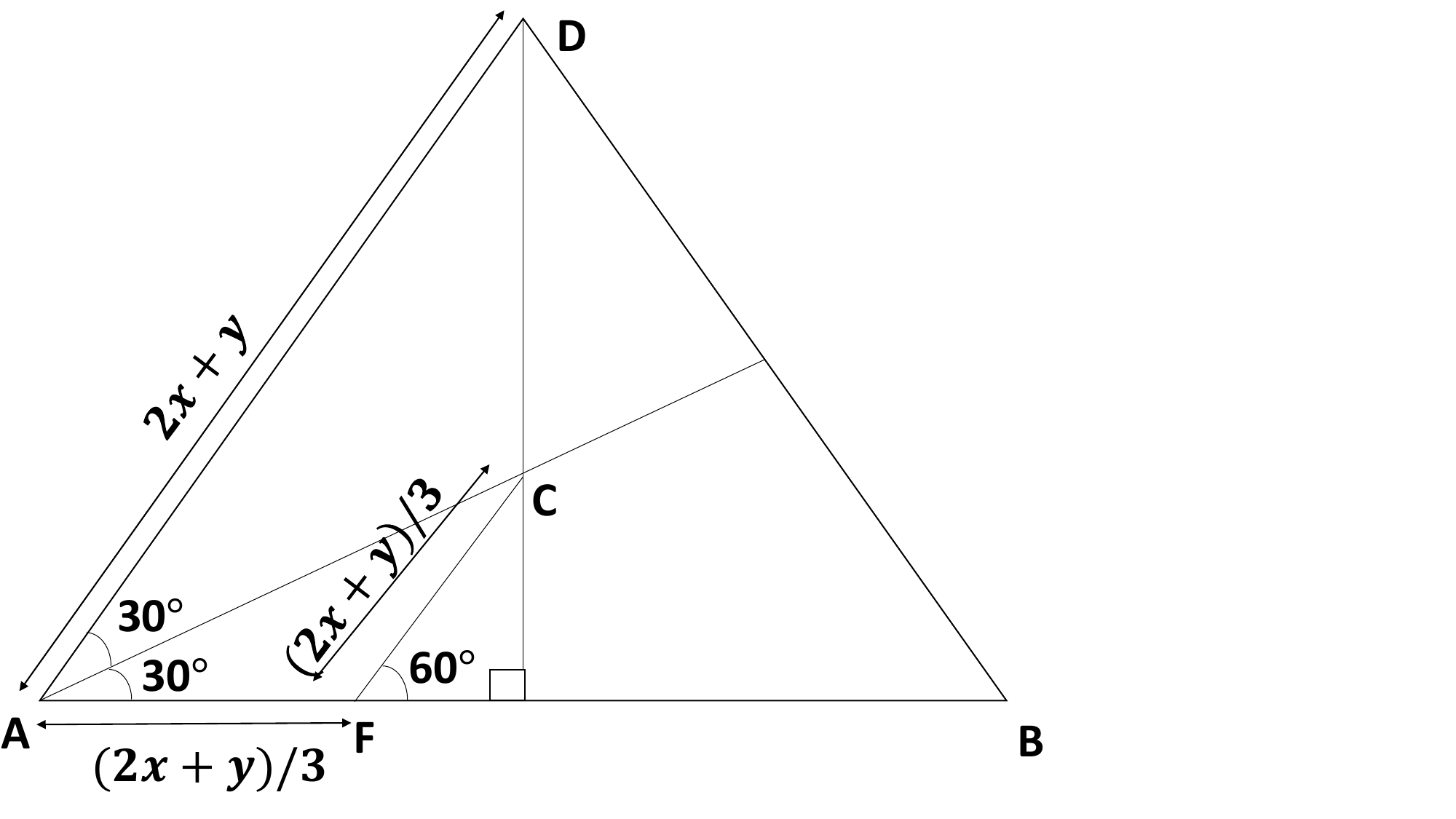}
  \caption{$C$ is the centroid of the triangle. $F$ trisects $\overline{AB}$. $\angle CFB = 60^\circ$.}
  \label{fig:cat3-hexagon-centroid}
\end{subfigure}
\caption{Category 3 hexagon with 3 SSLS of length $x$ and 3 LSLS of length $y$ inscribed in an equilateral triangle. $F$ is $(2x+y)/3$ distance from $A$ and $C$ is $(2x+y)/3$ distance from $F$.}
\label{fig:cat3-hexagon-relation-triangle}
\end{figure}

For a category 4 hexagon with side of length $x$, it is known that the centroid lies at a distance $x$ along the angle bisector of each vertex.

Thus for all four types of hexagons, a leader is chosen.

Regarding the running time, first  each vertex runs 6 instances of $\mathtt{Compare-Length(x)}$ in $O(L^2)$ rounds by Lemma~\ref{lem:running-time-compare-length} to decide which category the hexagon belongs to. For a category 1 hexagon, only an additional round is required until a unique leader is chosen. For a category 2 hexagon, each segment with $x$ nodes runs $\mathtt{Mid-Line}$, which takes $O(L^2)$ rounds by Observation~\ref{obs:running-time-mid-line}. Then each resulting snake runs $\mathtt{Snake-Movement(D,x)}$ in $O(L^2)$ rounds by Lemma~\ref{lem:running-time-snake-movement}. Then an additional round may be required for $\mathtt{MIS-Selection}$. For a category 3 hexagon, each segment performs the required calculations in $O(L^2)$ rounds and then the snake runs $\mathtt{Snake-Movement(D,g)}$ in $O(L^2)$ rounds by Lemma~\ref{lem:running-time-snake-movement} and possibly $\mathtt{MIS-Selection}$ for an additional round. For a category 4 hexagon, each snake with $x$ nodes runs $\mathtt{Snake-Movement(D,g)}$ in $O(L^2$) rounds. Thus, the total procedure takes $O(L^2)$ rounds for any particle.
\end{proof}

%----------------------------------------------------------------------------------------------

\subsection{Leader Election on a Spanning Tree}
\label{subsec:LE-spanning-tree}
% !TEX root = main.tex
 %please leave the above line untouched. Thanks, Billy.

Procedure $\mathtt{Spanning-Tree-Leader-Election}$ is presented to deterministically elect a unique leader when participating particles form a a spanning tree and have common chirality. Initially all participating particles have status \textbf{C}.

\alglanguage{pseudocode}
\begin{algorithm}
\caption{Spanning-Tree-Leader-Election, run by each particle $P$}
\label{prot:Spanning-Tree-Leader-Election}
\begin{algorithmic}[1]
	\State Let degree of $P$ be $\delta$
	\If {$P$ in state \textbf{C}}
		\If{$P$ received a \textit{leadership} message from $\delta$ adjacent particles}
			\State Change $P$'s status to \textbf{L}
		\ElsIf{$P$ received a \textit{leadership} message from $\delta-1$ adjacent particles}
			\State Change $P$'s status to \textbf{U}
			\State Send a \textit{leadership} message to $P$'s remaining adjacent particle.
		\EndIf
	\EndIf
\Statex
\end{algorithmic}
\end{algorithm}

The following theorem is given without proof as the correctness and running time follow trivially from the procedure description.

\begin{theorem}\label{the:sp-tree-le}
Procedure $\mathtt{Spanning-Tree-Leader-Election}$ run by particles forming a spanning tree of diameter $x$ results in exactly one leader being elected deterministically in $O(x)$ rounds.
\end{theorem}

\section{Leader Election}
\label{sec:leader-election-by-moving}
% !TEX root = main.tex
 %please leave the above line untouched. Thanks, Billy.

In this section, the main leader election algorithm is presented in detail. Initially, it is assumed that all particles have the same chirality, and it is shown in the next section how to remove this assumption.

First, an overview of the algorithm is given with the main theorem and supporting lemmas. Then each phase is explained in detail and relevant lemma(s) are proven in subsequent subsections.

\subsection{Algorithm Overview}
$\mathtt{Leader-Election-By-Moving}$ is a six phase deterministic algorithm run by a connected shape of particles with common chirality in order to elect a unique leader. Each phase of the algorithm is described in brief below with an extended description in the relevant subsection. As mentioned earlier, for the sake of convenience, the terms particle and node are used interchangeably when context makes the meaning clear. It is always the particle occupying a node that performs any action and when a particle is contracted, it only occupies one node.

The initial contracted configuration of $n$ particles forms a connected shape $G(0)$ at the beginning of round $0$ with all particles having status \textbf{C}. $H(0), K(0), F_1(0),$ and $F_2(0)$ are graphs at the beginning of round $0$ that are locally maintained by each particle and are initially empty.\footnote{For each graph, each particle maintains locally its own edges in the graph and whether it is in the graph or not. Each particle allots a constant amount of memory for each of the graphs $G, H, K, F_1, F_2$ and updates them as necessary when activated.} 
 Note that the round number is subsequently dropped, as it is apparent from context.  The algorithm has six phases. Graph $H$ is used throughout the algorithm for various purposes depending on the phase of the algorithm. Graph $K$ is a subgraph of $G$ that holds a spanning tree of all particles and is important for phase  6 of the algorithm. Graph $F_1$ is a forest of trees of all particles used throughout the algorithm. Graph $F_2$ is a forest of trees of a subset of the particles used only in phase 5 of the algorithm.

 Each particle $P$ maintains a phase counter, taking a constant number of bits, that it appends it to each message sent. If $P$ receives a message from another particle $Q$ in a different phase, $P$ does not process $Q$'s message until $P$ is in the same phase as $Q$.\footnote{It is trivial to return to a contracted configuration from the configuration the algorithm terminates in, so this ``7th phase'' is not described. Informally, it consists of particles that performed $\mathtt{Snake-Movement()}$ as part of $\mathtt{Convex-Polygon-Leader-Election}$ during phase 5 reversing their movements.}

\begin{enumerate}
\item \textit{Initialization:} Initially, all particles are contracted, have common chirality, and form a connected shape. At the end of this phase, every particle is contracted and each boundary particle has identified the type (inner/outer) of each of its boundaries. Furthermore, graph $H$ consists of a set of simple convex polygons, where two simple polygons may share the same semi-bridge particle.

Non-boundary particles change their phase to 2. Each boundary particle runs $\mathtt{Boundary-Detection}$ for each boundary $B$ it lies on to determine whether $B$ is an inner or outer boundary. Once $\mathtt{Boundary-Detection}$ terminates, all particles not on the outer boundary set $seg\_head = false$. Thus, there are $k$, $k \in \{1,2,3,6\}$, particles with $seg\_head = true$ located on the outer boundary. Call these seg-heads $P_1, P_2, \ldots, P_k$. If $k=1$, change $P_1$'s status to \textbf{L} and  broadcast (by simple flooding \cite{segall1983distributed}) a $final\_terminate$ message to other particles to terminate the algorithm and change their statuses to \textbf{U}.

Each particle that is not a bridge or semi-bridge particle adds itself and its edges to adjacent nodes to $H$. 
Semi-bridge particles add themselves and their non-bridge edges to $H$. 
Note that all particles are contracted at the end of this phase.

\item \textit{Spanning forest formation:} At the end of this phase, every particle has joined a tree in the spanning forest $F_1$, where every tree is rooted at one of the outer boundary nodes. Furthermore, each node knows its parent and children in the tree it joins.

Each outer boundary node $a$ becomes the root of a tree $T$ and sends messages to $a$'s neighbors to join $T$. Each node $b$ joins exactly one tree $T$, rejects requests from the other neighbors, and subsequently sends messages to $b$'s remaining neighbors to join $T$. This is done recursively until all nodes join some tree and then the phase terminates. Termination detection of a phase is coordinated by seg-heads $P_1$ to $P_k$. Thus a spanning forest of trees $F_1$ is formed with outer boundary particles as roots of the trees. Note that all particles are contracted at the end of this phase.

\item \textit{Convexification:} The subgraph $H$, induced by removing bridge particles from the shape, is a collection of polygons.  Each outer boundary particle $P$ w.r.t.\ $H$ that is a concave vertex and not a semi-bridge particle expands towards the outer boundary along $P$'s angle bisector while coordinating the pulling of $P$'s tree with it. $P$ occupying node $b$ and moving to node $c$ completes one step of convexification when it has moved to node $c$ and all particles in the tree rooted at $P$ in $F_1$ are in a contracted state. Convexification is performed repeatedly by particles until no more steps of convexification are possible.

At the same time, each seg-head $P_i$ ($1 \leq i \leq k$) continuously checks its segment for any concave vertices in $H$. If none are found, the $k$ seg-heads coordinate to terminate this phase. All particles previously in $H$ update their edges in $H$ to reflect current connections to other particles. Furthermore, bridge and semi-bridge particles add themselves and their bridge edges to graph $K$. Note that all particles are contracted at the end of this phase.

In this phase, there is a possibility for a reset to occur causing all particles to reset their states and start the algorithm afresh from phase one. There are two types of resets, \emph{type 1} and \emph{type 2}. A type 1 reset is triggered by a boundary particle that moved in some direction $D$ to node $b$ in one step of convexification, finding that the node adjacent to $b$ in direction $D$ is occupied. This reset occurs because the length of the boundary might have decreased. The trigger condition for a type 2 reset occurs when a particle $P$ which is initially a semi-bridge particle, a bridge particle, or an outer boundary particle stops being one. Both resets reflect a change in the particles occupying the outer boundary, possibly resulting in particles previously with $seg\_head = true$ no longer lying on the outer boundary. The exact procedure for each type of reset is described in detail in the relevant subsection. Note that the reset procedure ensures that all particles are contracted prior to returning to phase one.

\item \textit{De-sharpification:} In this phase, certain particles in $H$ remove themselves recursively until only convex polygons and two-node lines remain in $H$. Consider a particle $P$ in $H$. If $P$ is not a semi-bridge particle and is a sharp vertex w.r.t.\ the outer boundary in $H$, then $P$ removes itself from $H$. If $P$ is a semi-bridge particle and its occupied adjacent nodes are located at ports $x, x+1, x+3,$ and $x+4$ ($\mod 6$) for some positive integer value of $x$, then $P$ removes itself from $H$.

At the same time, each seg-head $P_i$ ($1 \leq i \leq k$) continuously checks its segment for any sharp vertices in $H$. If none are found, $P_i$ coordinates with the other $k-1$ seg-heads to terminate the phase.  The induced subgraph $H$ at the end of the phase is a set of convex polygons without sharp vertices and lines consisting of 2 nodes. Note that all particles are contracted at the end of this phase.

\item \textit{Leader election on individual polygons and spanning tree formation:} This phase consists of two stages. In stage one, each convex polygon and each line in $H$ elects a unique polygon leader using $\mathtt{Convex-Polygon-Leader-Election}$ and $\mathtt{MIS-Selection}$, respectively. In stage two, each particle $P$ chosen as a polygon leader in stage one, acts as a root and forms a tree that spans its connected component of $G\setminus K$. The nodes in $K$ that are reachable from $P$ over $G\setminus K$ are leaves of $P$'s tree.
  Call this forest of polygon leaders rooted trees $F_2$.

At the same time, each seg-head $P_i$ ($1 \leq i \leq k$) continuously checks if its segment satisfies the termination condition, described in detail in the relevant subsection. Once each seg-head $P_i$'s segment satisfies the termination condition, the $k$ seg-heads coordinate with each other and terminate this phase.

At the end of this phase, $K$ is updated to contain all particles in graph $G$ with edges restricted to bridge edges and edges from $F_2$.

\item \textit{Leader election on a spanning tree:} Each particle participates in $\mathtt{Spanning-Tree-Leader-Election}$ on the graph $K$.\footnote{Specifically, the degree of each particle $P$ is calculated based on which particles are adjacent to $P$ w.r.t. $K$.} Once a particle $P$ changes its status to \textbf{L}, $P$ broadcasts a $final\_terminate$ message by flooding along $K$. This results in one particle, the leader, having status \textbf{L} and the remaining particles having status \textbf{U} when the algorithm terminates.
\end{enumerate}

The following lemmas apply to the algorithm and are proven in the respective subsections.

\begin{lemma}\label{lem:main-alg-phase-one}
	Phase 1 terminates in $O(L_m^2)$ rounds, where $L_m$ is the length of the largest boundary of the shape, resulting in each boundary particle knowing what type each of its boundaries is and $k$, $k \in \{1,2,3,6\}$, particles, $P_1, P_2, \ldots, P_k$, lying on the outer boundary with $seg\_head = true$. Furthermore, at the end of the phase, $H$ consists of a set of simple convex polygons, where two simple polygons may share the same semi-bridge particle. If $k=1$, the algorithm terminates with one particle as leader in an additional $O(n)$ rounds.
\end{lemma}

\begin{lemma}\label{lem:main-alg-phase-two}
	Phase 2 terminates in $O(n)$ rounds, resulting in a disjoint forest of trees $F_1$ covering every particle.
\end{lemma}

\begin{lemma}\label{lem:main-alg-phase-three}
	Phase 3 terminates in $O(Ln)$ rounds, resulting in either a reset or a graph $H$ containing a set of simple convex polygons, where two simple polygons may share the same semi-bridge particle.
\end{lemma}

\begin{lemma}\label{lem:num-resets}
	There can be at most $L$ resets occurring in phase 3, where $L$ is the length of the outer boundary of the original shape.
\end{lemma}

\begin{lemma}\label{lem:main-alg-phase-four}
	Phase 4 takes $O(n)$ rounds to complete, resulting in $H$ containing only a set of lines consisting of 2 nodes and convex polygons without sharp vertices.
\end{lemma}

\begin{lemma}\label{lem:main-alg-phase-five}
	Phase 5 terminates in $O(L^2 + n)$ rounds resulting in $K$ containing all particles and forming a spanning tree.
\end{lemma}

\begin{lemma}\label{lem:main-alg-phase-six}
	Phase 6 terminates in $O(n)$ rounds resulting in a unique leader with status \textbf{L} being chosen and all other nodes having status \textbf{U}.
\end{lemma}

Combining the above lemmas together, we get the total running time and correctness of the main algorithm.

\begin{theorem}\label{the:le-by-moving}
When algorithm $\mathtt{Leader-Election-By-Moving}$ is run by $n$ particles in a contracted configuration, it elects a unique leader deterministically and terminates in $O(L  n^2)$ rounds, where $L$ is the number of particles on the outer boundary of the original shape.
\end{theorem}

\begin{proof}[Proof Sketch]
From Lemmas~\ref{lem:main-alg-phase-one},~\ref{lem:main-alg-phase-two}, and~\ref{lem:main-alg-phase-three}, the combined running time of one iteration of phases 1 to 3 is $O(n^2)$ rounds since $L_m = O(n)$. There can be at most $O(L)$ iterations of phases 1 to 3, by Lemma~\ref{lem:num-resets}. Adding in the running times of phases four, five, and six from Lemmas~\ref{lem:main-alg-phase-four},~\ref{lem:main-alg-phase-five}, and~\ref{lem:main-alg-phase-six}, it is clear that the total running time of the algorithm is $O(Ln^2)$ rounds.

The correctness directly follows from Lemma~\ref{lem:main-alg-phase-six}.
\end{proof}

\subsection{Initialization}
Recall that a particle can locally determine whether it is a boundary particle or not by checking that at least one adjacent node to it is unoccupied. Non-boundary particles change their phase to phase 2 and wait. The boundary particles run procedure $\mathtt{Boundary-Detection}$ for each of their boundaries. At the end of the procedure, all participating particles know if they are on the outer boundary or an inner one. Furthermore, for the outer boundary, there remain exactly $k \in \{1,2,3,6\}$ particles with $seg\_head = true$ and each of them knows the value of $k$. If more than 1 seg-head remains on the outer boundary, the algorithm proceeds with the remaining phases, else this seg-head becomes the leader and broadcasts its election to all particles.

For each boundary $B$, some $k_B$ particles on that boundary set $seg\_head = true$ and know the value of $k_B$. Each of the $k_B$ seg-heads then sends a $terminate$ message to the remaining $k_B-1$ seg-heads on $B$ and waits to receive $k_B-1$ such $terminate$ messages. Then the seg-head sends a message to all particles in its segment to change their phase counter to 2. If $B$ is an inner boundary, each of the $k_B$ seg-heads set $seg\_head = false$.

At this time, all particles except bridge particles add themselves and any non-bridge edges to~$H$.

The lemma about this phase is proven below.

\begin{proof}[Proof of Lemma~\ref{lem:main-alg-phase-one}]
From Theorem~\ref{the:boundary-detection}, we see that each boundary particle $P$ running $\mathtt{Boundary-Detection}$ detects the type of each of its boundaries $B$ in $O(L_m^2)$ rounds. Furthermore, if $P$ is one of the $k$ seg-heads on boundary $B$ with $seg\_head = true$, it knows the value of $k$.

Furthermore, $H$ is the shape induced by removing bridge particles from the shape. Since only bridge particles can have boundary count 3, the resulting shape in $H$ consists of a set of simple convex polygons, where two simple polygons may share the same semi-bridge particle.
\end{proof}

\subsection{Spanning Forest Formation:}
Recall that at the end of this phase, every particle has joined a tree in the spanning forest $F_1$, where every tree is rooted at one of the outer boundary nodes.

The following procedure is run by each non-outer boundary node. Let $b$ be a node that has not yet joined a tree but has received requests to join trees from a subset $\mathcal{S'}$ ($|\mathcal{S'}| > 0$) of its neighboring occupied nodes $\mathcal{S}$. The node $b$ sends messages to nodes in $\mathcal{S} \setminus \mathcal{S'}$ asking them to become $b$'s children. Once $b$ receives an accept/reject message from each node in $\mathcal{S} \setminus \mathcal{S'}$, $b$ chooses one node in $\mathcal{S'}$ to be its parent, and sends accept/reject messages to each node in $\mathcal{S'}$. 

For each node $a$ in a segment on the outer boundary, $a$ sends a request to all its neighbors for them to become $a$'s children. If $a$ receives a join request from some other node, it rejects it. When $a$ receives an accept/reject message from all its neighbors and an acknowledgement message from its successor in the segment, if any, $a$ sends an acknowledgement message to its  successor node in the segment.\footnote{Notice that the tail of the segment has no successor within the segment and so it sends an acknowledgement to its successor once its receives an accept/reject message from all of its neighbors.}

Termination detection of this phase is coordinated by seg-heads $P_1$ to $P_k$ as follows. Once some $P_i$ receives an acknowledgement message from its successor, if any, and accept/reject messages from its neighbors, it sends a $terminate$ message to the remaining $k-1$ other seg-heads.\footnote{Recall that seg-head $P_i$ can send a message to a seg-head $P_{((i+x -1) \mod k) +1}$ by embedding $x$ into the message as well as a counter $y$ initialized to $0$. Every seg-head that receives the message increments $y$ until $y=x$, indicating that the message has reached the intended target seg-head.} Once $P_i$ receives $k-1$ $terminate$ messages in turn, $P_i$ sends a message to all particles of $F_1$ reachable from $P_i$ informing them to change to phase 3.

We now prove the lemma about this phase.

\begin{proof}[Proof of Lemma~\ref{lem:main-alg-phase-two}]
Consider the height $h$ of the longest tree rooted at some outer boundary node $b$, belonging to segment $s$, that may be formed in $F_1$. It is clear that $h = O(n)$. Thus it takes $O(n)$ rounds to add nodes at each level of the tree, an additional $O(n)$ rounds to send an acknowledgement to the root of the tree that the tree is built, and a further $O(L)$ rounds to inform the head of $s$. Finally, an additional $O(L)$ rounds are required for $s$ to coordinate with the other segments to terminate the phase, resulting in $O(n)$ rounds needed by each particle to terminate the phase.
\end{proof}

\subsection{Convexification}
One step of convexification is now described. Consider a particle $P$ that is not a semi-bridge particle, occupying some node $b$, lying on the outer boundary $B$. Let the tree rooted at $b$, formed in the previous phase as part of $F_1$, be $T$. Now, assume $b$ is a concave vertex w.r.t.\ $B$. Consider the angle bisector of $b$ w.r.t. $B$ and the unique node $c$ that is adjacent to $b$ and lies in the outer face. Let $D$ be the direction of $c$ w.r.t. $b$. First $P$ expands in direction $D$ to node $c$. Then $P$ contracts and pulls one of its children $Q$ into $b$. This process is repeated recursively until a final particle $R$ without a child of its own in $T$ is pulled into some node $d$. $R$ then contracts into $d$ and sends an acknowledgement message to its parent in $T$. This acknowledgement message is  sent recursively through the parents in $T$ until it reaches $P$. Figure~\ref{fig:convexification} illustrates an example of the convexification process.

\begin{figure}
\begin{subfigure}{0.5\textwidth}
  \centering
  \includegraphics[page=30,width=.8\linewidth]{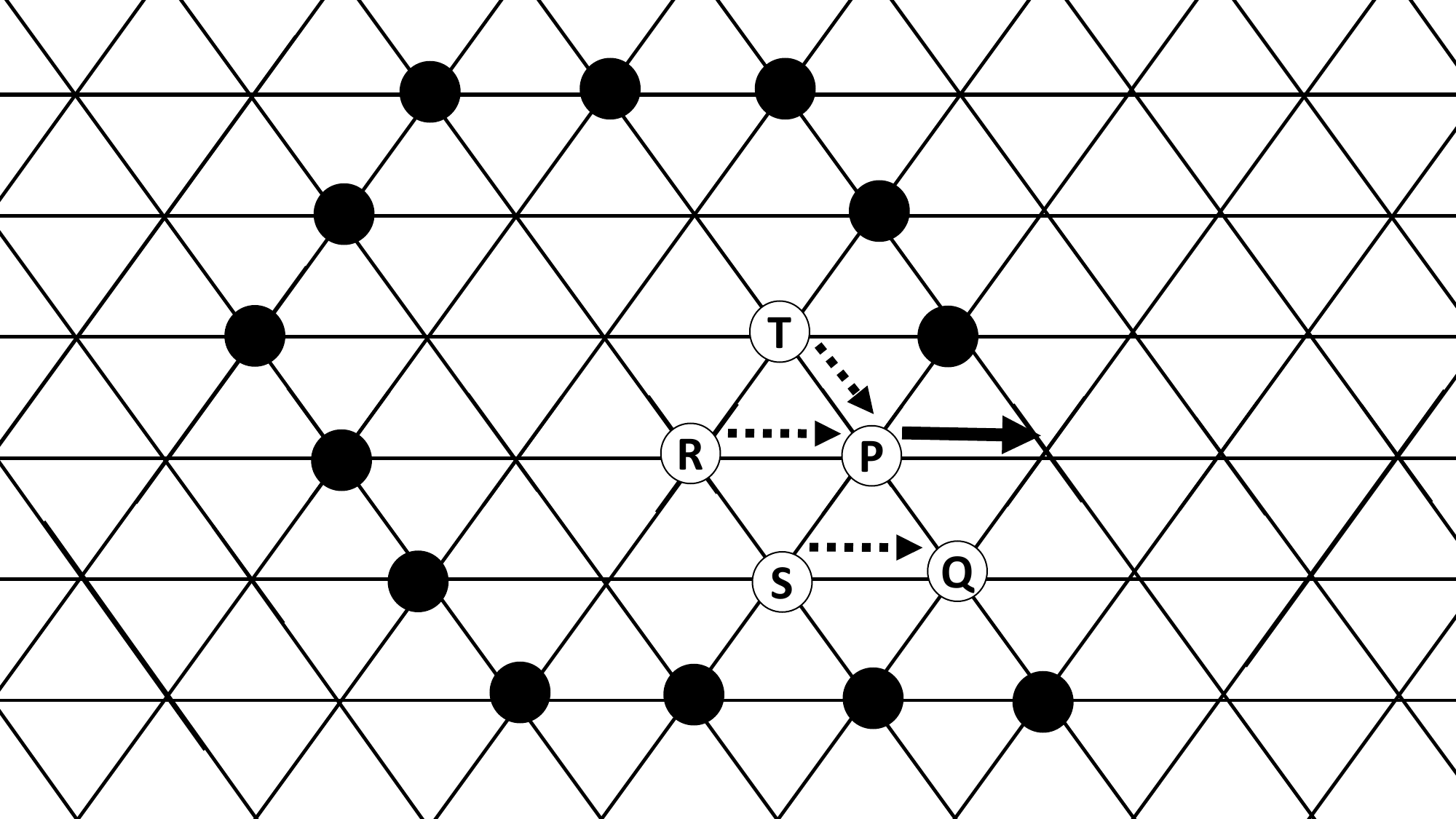}
  \caption{$P$ is a concave vertex w.r.t.\ the outer boundary.}
  \label{fig:convexification-step1}
\end{subfigure}%
\begin{subfigure}{0.5\textwidth}
  \centering
  \includegraphics[page=31,width=.8\linewidth]{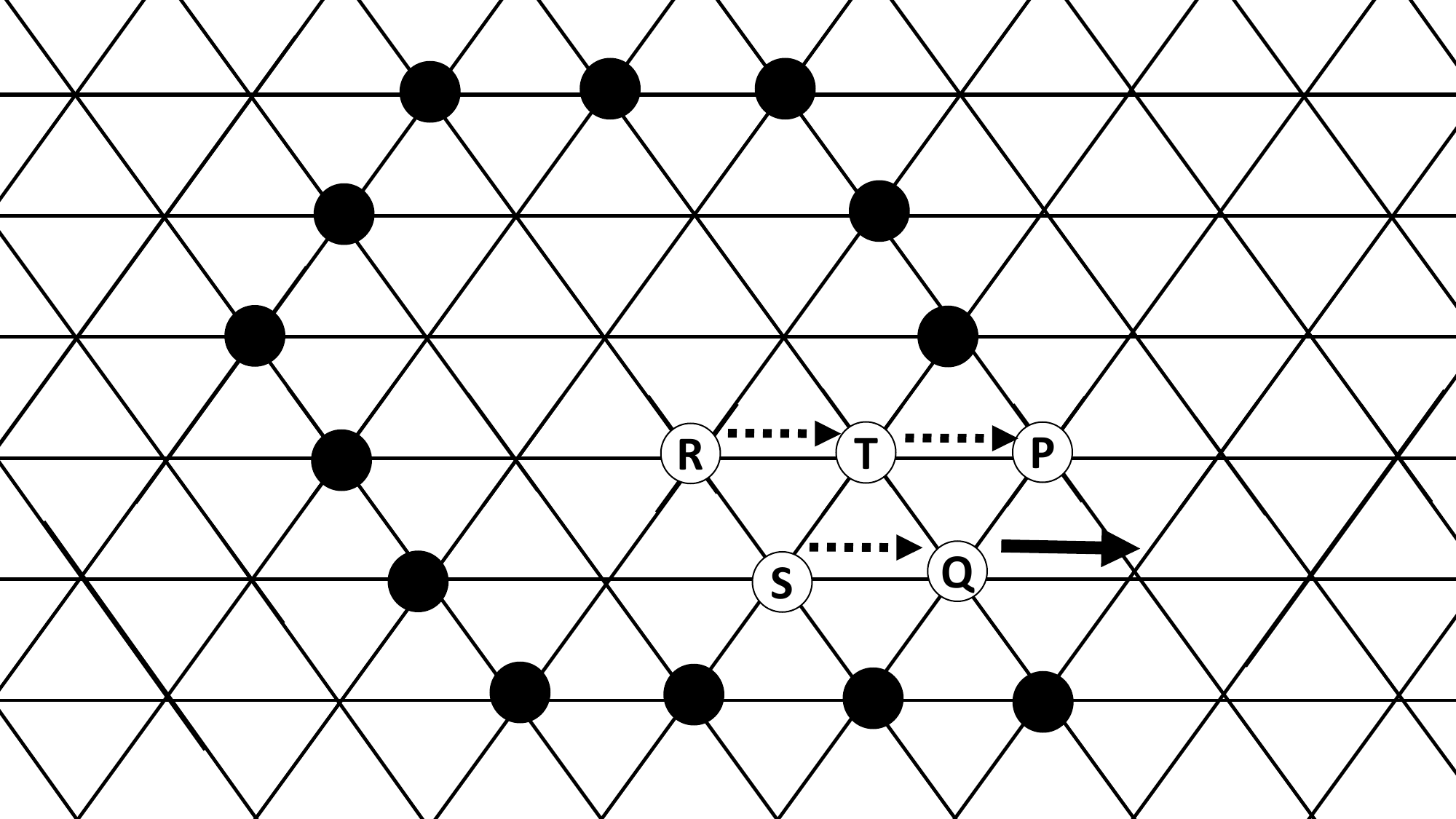}
  \caption{$Q$ is a concave vertex w.r.t.\ the outer boundary.}
  \label{fig:convexification-step2}
\end{subfigure}
\begin{subfigure}{0.5\textwidth}
  \centering
  \includegraphics[page=32,width=.8\linewidth]{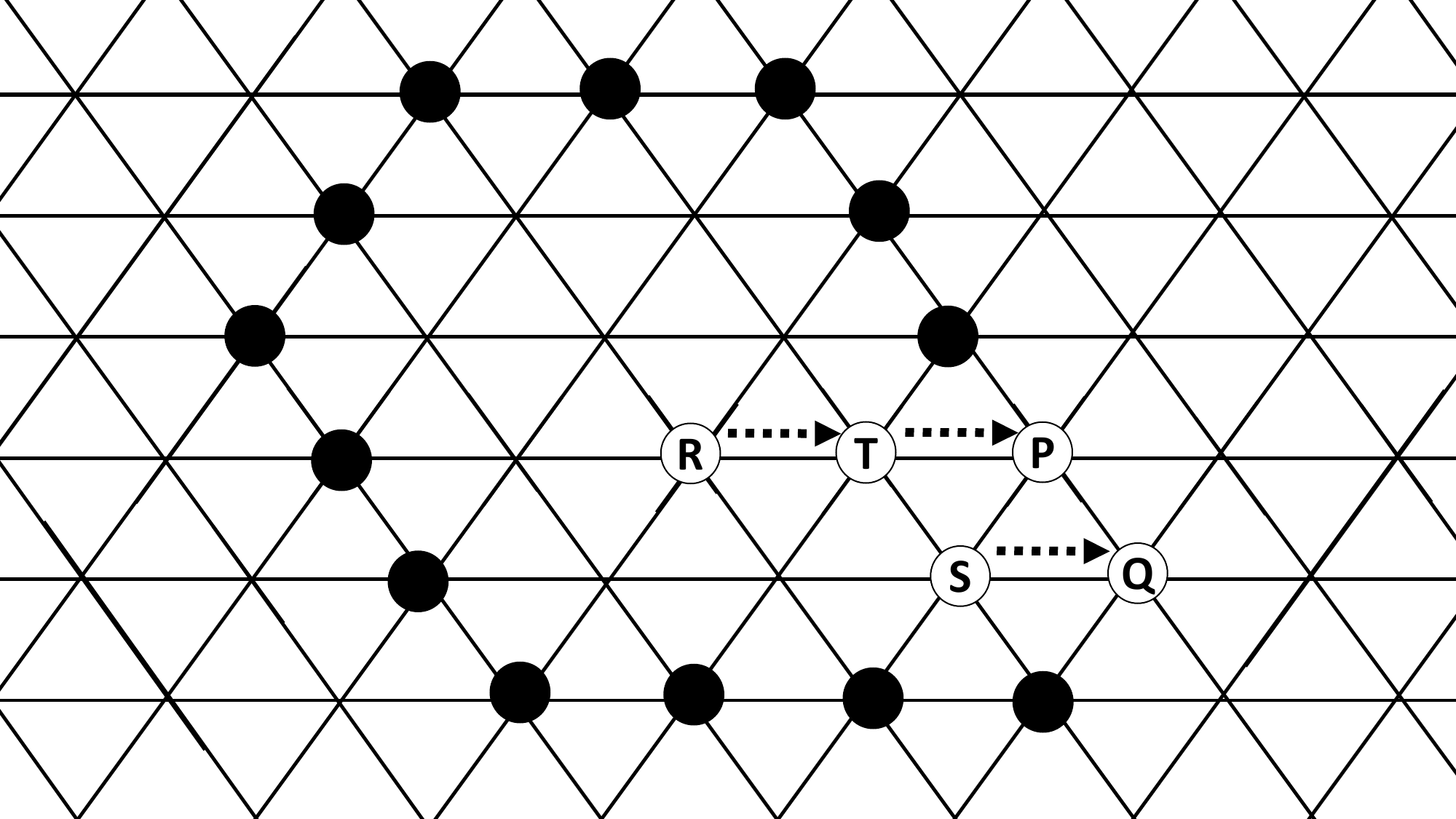}
  \caption{Convexification done.}
  \label{fig:convexification-step3}
\end{subfigure}
\caption{Convexification of a non-convex polygon with outer boundary particles $P$ and $Q$. $R$ and $T$ are part of $P$'s tree. $S$ is part of $Q$'s tree.}
\label{fig:convexification}
\end{figure}

The procedure run by the seg-heads $P_i$, ($1 \leq i \leq k)$, on the outer boundary to detect termination of this phase is now described. Each seg-head $P_i$ ($1 \leq i \leq k$) continuously checks its segment for any concave vertices in $H$. If none are found, $P_i$ sends a message to all trees rooted at nodes in $P_i$'s segment to move as necessary so that all particles are in a contracted state. Once $P_i$ receives an acknowledgement that this is complete, $P_i$ sends a $terminate$ message to the other $k-1$ seg-heads. Once $P_i$ receives a $terminate$ message from the other seg-heads, $P_i$ sends a message to all particles in $F_1$ belonging to its segment or the subtree rooted at a node in its segment, indicating that this phase of the algorithm is over. 

The type 1 reset procedure is now described. Consider a particle $P$ occupying outer boundary node $b$ with predecessor $c$ and successor $d$. Let it expand in direction $D$ to node $e$. Let $f$ be the node adjacent to $e$ in the direction $D$. After $P$ completes one step of convexification, if $f$ is already occupied, then the outer boundary has just been partitioned into two boundaries at least one of which is now an inner boundary. $P$ generates a $reset$ message which $c$ propagates to one boundary and $d$ propagates to the other. Let $g$ be a node that is on the outer boundary at the start of this phase. When $g$ receives a $reset$ message, it sends a message to all particles in the tree rooted at $g$ to move as necessary so that all particles are in a contracted state. Once $g$ receives an acknowledgement that this is done, it sends a message to all particles in the tree rooted at $g$ telling them to reset their states and change their phase to one. Once $g$ receives an acknowledgement that this is done, $g$ passes on the $reset$ message along the boundary. Once each previous outer boundary particle $g$, for each outer boundary it lies on, sends and receives a $reset$ message or receives a $reset$ message from its predecessor and successor nodes, $g$ restarts the algorithm from phase one. The particle that generated the $reset$ restarts the algorithm once it sends out the $reset$ message to its two neighboring nodes.

The procedure for type 2 resets is now described. Consider a particle $P$ occupying node $b$ which is either a bridge particle, semi-bridge particle, or outer boundary particle. When $P$ changes its nature due to one of its adjacent previously unoccupied nodes becoming occupied, a type 2 reset occurs. In all three cases, $P$ generates the $reset$ message and sends it to particle(s) that caused $P$ to change. These particles in turn send the $reset$ messages to their successor and predecessor nodes along the outer boundary. The $reset$ message itself, once received by a node, operates in the same manner as described for a type 1 reset. Similarly, once a boundary particle other than $P$, for each of the outer boundaries it lies on, sends and receives a $reset$ message or receives a $reset$ message from its predecessor and successor nodes, it restarts the algorithm from phase one.\footnote{Recall that a contracted particle $P$ may lie on at most 3 boundaries with corresponding successor and predecessor nodes, and each of these boundaries may be the outer boundary. If $P$ does not generate a $reset$ message, then for each of these boundaries, either $P$ receives $reset$ messages from both its predecessor and successor or else it receives a $reset$ message from one and passes it on to the other. By checking for either case to occur along each boundary, we ensure that the $reset$ message does not circulate around the boundary forever.} When $P$ passes the $reset$ message to all adjacent particles that influenced it, $P$ restarts the algorithm.

A proof sketch is provided for Lemma~\ref{lem:main-alg-phase-three} and a proof is given for Lemma~\ref{lem:num-resets} below.

\begin{proof}[Proof Sketch of Lemma~\ref{lem:main-alg-phase-three}]
	We first bound the number of nodes a particle may move through during convexification, assuming no reset occurs, and then bound the time it takes to move through one node during convexification.
	
	Graph $H$ consists of a disjoint set of polygons. Consider one such polygon $A$ of outer boundary length $x$. Now, circumscribe this polygon with a regular hexagon $B$ with length of side $y$. It can be easily seen that the polygon obtained after convexification of $A$ fits in or is equal to $B$. Furthermore, it can be seen that $y \leq x$ and the outer boundary length of $B$ is $6y$. Thus, any particle in $A$ does not move more than $6y$ nodes to reach its final position in $B$. Notice that $x = O(L)$ and thus $y = O(L)$. Thus a given particle on the boundary does not move more than $O(L)$ nodes during convexification.
	
	Let us now calculate the time it takes to complete one step of convexification, i.e. have a particle move to a new node and subsequently have all particles in its tree in $F_1$ move into a contracted state. The height of any tree $T$ in $F_1$ is upper bounded by $n$. To move one step, first the root of $T$ performs an expand. Subsequently, each particle in one branch of $T$, from root to leaf must perform a pull, taking $O(n)$ rounds. Then an acknowledgement message is sent to the root in $O(n)$ rounds. Thus, any particle takes $O(Ln)$ rounds to complete the phase.
	
	If a reset occurs, then it takes $O(n)$ rounds to propagate the reset message to all nodes. If the phase terminates successfully, it takes $O(n)$ rounds for the $P_i$, $1 \leq i \leq k$, seg-heads to coordinate the $terminate$ messages. Thus the running time of this phase is $O(Ln)$ rounds.
	
	Clearly, if the termination detection of the phase is successful, i.e. no reset occurs, then each simple polygon in $H$ has successfully been convexified. Note that two simple polygons may share the same semi-bridge particle.\footnote{As an example of such a semi-bridge particle, consider the formation in the top-right of Figure~\ref{fig:bridge-semibridge-particles}.} Thus the resulting graph $H$ is a set of simple convex polygons.
\end{proof}

\begin{proof}[Proof of Lemma~\ref{lem:num-resets}]
	Let us first show that if the original shape has an outer boundary of length $L$, at most $O(L)$ type 2 resets, each of which does not increase the length of the boundary. Subsequently, it is shown that with each type 1 reset, the length of the boundary decreases by at least 1. Thus, there can be at most $O(L)$ resets.
	
	A bridge particle $P$ occupying node $a$ with adjacent occupied node $b$ changes to a semi-bridge particle when an unoccupied node $c$ adjacent to both $a$ and $b$ becomes occupied. If enough nodes are occupied, $P$ can change into just an outer boundary particle or even a non-boundary particle. For each such node $c$, this results in the outer boundary of the shape gaining at most one new node. A similar principle holds for semi-bridge particles and outer boundary particles. There at most $6$ unoccupied nodes surrounding any contracted particle and at most $L$ bridge, semi-bridge, and contracted outer boundary particles in the original shape. Thus, there can be at most $O(L)$ type 2 resets that occur for each outer boundary particle. Furthermore, notice that the particles occupying adjacent nodes to $P$ are themselves outer boundary particles. Thus, the length of the outer boundary does not increase with each type 2 reset and may in fact decrease.	
	
	Now, each type 1 reset occurs when two outer boundary particles that are neither predecessors nor successors to each other along the outer boundary become adjacent to each other. One type 1 reset splits the outer boundary into an inner boundary and an outer boundary.  Assuming no type 2 resets are also triggered before the reset completes, the outer boundary length decreases by at least $1$ after a type 1 reset. Thus there can be at most $O(L)$ type 1 resets. Thus the total number of type 1 and type 2 resets is bounded by $O(L)$. 
\end{proof}

\subsection{De-sharpification}
The procedure to detect termination of the phase 4 is described. Each seg-head $P_i$ ($1 \leq i \leq k$) on the outer boundary continuously checks its segment for any sharp vertices in $H$. If none are found, $P_i$ sends a $terminate$ message to the other seg-heads. Once $P_i$ receives a $terminate$ message from the other seg-heads, $P_i$ sends a message to all particles in its segment and the subtrees rooted at nodes in the segment, indicating that this phase of the algorithm is over.

The lemma concerning this phase is proven below.

\begin{proof}[Proof of Lemma~\ref{lem:main-alg-phase-four}]
First, let us bound the running time and then prove correctness. At the beginning of this phase, $H$ contains a set of simple convex polygons, where two simple polygons may share the same semi-bridge particle. Consider two such simple convex polygons sharing the same semi-bridge particle $P$. This means that $P$ must have occupied adjacent nodes located at ports $x, x+1, x+3,$ and $x+4$ ($\mod 6$) for some positive integer value of $x$. Removing all such particles $P$ from $H$ results in $H$ containing a set of convex polygons and lines with two nodes. Now, from the remaining convex polygons in $H$, sharp vertices can identify themselves locally and remove themselves from $H$. It is clear to see that for each polygon, in $O(n)$ rounds, sharp vertices remove themselves until a convex polygon without sharp vertices is formed or else a line with two nodes is formed. Subsequently, the seg-heads $P_i$, $1 \leq i \leq k$, take $O(n)$ rounds to propagate the $terminate$ message to each other and update the phase for particles in the forest $F_1$.

Regarding correctness, the phase must terminate only when all sharp vertices w.r.t.\ the outer boundary are removed from $H$. It is easy to see that if a sharp vertex $b$ w.r.t.\ the outer boundary is removed and results in another sharp vertex $c$ being formed, then $c$ lies on one of the $k$ segments on the outer boundary. Thus, until all sharp vertices w.r.t.\ the outer boundary are removed, there exists at least one seg-head $P_i$ on whose segment this sharp vertex lies. Thus, only once all sharp vertices are removed from $H$ do particles in $F_1$ receive messages to change the phase.
\end{proof}

\subsection{Leader Selection for Each Polygon}
This phase consists of two stages. In stage one, each particle $P$ in $H$ checks if the node it occupies has a boundary count $3$ w.r.t.\ the outer boundary in $H$. If so, $P$ runs $\mathtt{MIS-Selection}$. Else, $P$ runs $\mathtt{Convex-Polygon-Leader-Election}$. Call the particle that joined the MIS/was chosen as a leader a \emph{polygon leader}.

In stage two, each polygon leader $P$ sends a message to adjacent nodes to become its children in the tree that spans its connected component of $G\setminus K$ with nodes in $K$ that are reachable from $P$ over $G\setminus K$ as leaves of $P$'s tree. Consider a node $b$ in $P$'s connected component of $G\setminus K$ not already part of the tree with $\mathcal{S}$ neighboring occupied nodes and that received messages from $\mathcal{S'}$ ($|\mathcal{S'}| > 0$) nodes to become their child. The node $b$, once activated, chooses one of the $\mathcal{S'}$ nodes arbitrarily as its parent and rejects the others. If $b$ is in $K$, it sends these accept and reject replies to the nodes in $\mathcal{S'}$ when activated. If $b$ is not a semi-bridge particle, it first sends a message to the remaining $\mathcal{S} \setminus \mathcal{S'}$ nodes, asking them to become its children. Once these nodes reply, $b$ in turn sends it accept/reject replies to nodes in $\mathcal{S'}$. Once a polygon leader receives a reply from each of its adjacent nodes, it sends a $freeze$ message to all its children in the tree, which is in turn propagated throughout the tree.

The termination condition for each segment is now described. Consider a semi-bridge particle $Q$ in $K$, at the beginning of the phase, occupying node $a$ and the set $\mathcal{S''}$ of adjacent occupied nodes not including nodes that form bridge edges with $a$. $\mathcal{S''}$ forms either one or two connected components. If $\mathcal{S''}$ consists of two nodes $b$ and $c$ such that both edges $(a,b)$ and $(a,c)$ lie on an inner boundary and outer boundary, consider $b$ and $c$ to belong to the same connected component. When $Q$ has received a $freeze$ message from a node in each of these connected components, $Q$ is said to be a \emph{frozen} semi-bridge particle.

Each seg-head $P_i$ ($1 \leq i \leq k$) on the outer boundary continuously checks its segment for any semi-bridge particles that are not frozen. If none are found, $P_i$ sends an update message to all particles in the subtrees of $F_1$ rooted at nodes in $P_i$'s segment, informing them to add themselves and any edges of $F_2$ to $K$. Once $P_i$ receives an acknowledgement message that the above is done, it sends a $terminate$ message to the other $k-1$ seg-heads. Once $P_i$ receives a $terminate$ message from the other seg-heads, $P_i$ sends a message to all particles in its segment and the subtrees of $F_1$ rooted at nodes in the segment, indicating that this phase of the algorithm is over. %Once $P_i$ receives an acknowledgement that the above process is done, it changes its status to \textbf{N}.

The lemma for this phase is proven below.

\begin{proof}[Proof of Lemma~\ref{lem:main-alg-phase-five}]
The running time is the sum of running times of $\mathtt{MIS-Selection}$, $\mathtt{Convex-Polygon-Leader-Election}$, tree formation, and coordinating the termination detection of the phase. The latter two take $O(n)$ rounds, and from Observation~\ref{obs:mis-selection} and Theorem~\ref{the:convex-poly-le}, the running time follows.

It can be seen that semi-bridge particles always lie on an outer boundary and thus the phase does not terminate until all semi-bridge particles are frozen. For a given convex polygon with semi-bridge particles, once all these semi-bridge particles are frozen, the tree spanning all particles of that polygon has been created. Once all semi-bridge particles, if any, in the shape are frozen, a forest of trees spanning all such convex polygons has been created. The addition of bridge edges and bridge and semi-bridge particles to these trees results in a tree spanning the entire shape.  Thus, the final $K$ that results is a spanning tree.
\end{proof}

\subsection{Leader Election on a Spanning Tree}
All particles have status \textbf{C} at the start of this phase and are in $K$. Furthermore, the edges of $K$ are restricted so that $K$ forms a spanning tree. Thus particles can participate in $\mathtt{Spanning-Tree-Leader-Election}$ resulting in a unique leader being chosen.

The lemma for this phase is proven below.

\begin{proof}[Proof of Lemma~\ref{lem:main-alg-phase-six}]
The correctness follows directly from Theorem~\ref{the:sp-tree-le}. The diameter of the tree is upper bounded by $n$ and broadcast takes $O(n)$ rounds, thus the running time until termination of this phase is $O(n)$ rounds.
\end{proof}

%----------------------------------------------------------------------------------------------
\section{Chirality Agreement}
\label{sec:chir-agreement}
% !TEX root = main.tex
 %please leave the above line untouched. Thanks, Billy.

In this section, the procedure $\mathtt{Chirality-Agreement}$ is described in detail. Consider $n$ contracted particles, with binary flags $ready$ and $chir\_ready$ initially set to $false$, forming a connected shape with the length of maximum boundary being $L_{\max}$. The particles run $\mathtt{Chirality-Agreement}$ and terminate in $O(L_{\max}^2 + n)$ rounds, resulting in all particles agreeing on the same chirality and forming the original shape. 

The procedure works in five phases, briefly explained below. Initially, each boundary particle $P$ identifies each boundary $B$ it lies on and $P$'s predecessor/successor nodes w.r.t.\ $B$. The next phase consists of $P$ identifying, for each of its boundaries $B$, whether $P$'s chirality agrees with that of its neighbors w.r.t.\ $B$. In the third phase, for every boundary $B$, the particles on $B$ coordinate to agree on their chirality.\footnote{A particle $P$ may maintain different chiralities for each of its boundaries in phase two. Each particle eventually chooses one overall chirality in phase four.} Furthermore, each boundary particle $P$ identifies, for each boundary $B$ it lies on, whether $B$ is an inner or outer boundary. Phase 4 consists of particles on the outer boundary with chirality $C$ informing other particles to take on the chirality $C$. The final phase is used by particles to detect termination of the procedure. A detailed explanation of the procedure is given below.

\begin{enumerate}
\item \textit{Initialization:} Each particle $P$ determines whether it is a boundary particle or not. If $P$ is not a boundary particle, $P$ sets $ready = true$. For each boundary $B$ that $P$ lies on, $P$ marks the ports to its predecessor and successor nodes w.r.t.\ $B$ and sets flag $ready = true$.

$P$ moves to the next phase once it and all surrounding particles have $ready = true$.

\item \textit{Chirality identification of neighbors along each boundary:} Consider two neighboring boundary nodes $b$ and $a$ occupied by boundary particles $P$ and $Q$ respectively. Assume that $a$ (and $Q$) is the predecessor of $b$ w.r.t. some boundary $B$, according to $P$'s chirality. Note that $Q$'s chirality may or may not be the same, so it may or may not consider $P$ as the successor. 
   Moreover, both $P$ and $Q$ may have multiple boundaries.  
     How can $P$ tell $Q$ that $Q$ is the predecessor w.r.t. $B$ (rather than to some other boundary they may share)?
     Note that for $B$ to be a boundary, nodes $b$ and $a$ must both neighbor exactly unoccupied node $c$ on boundary $B$. 
     Particle $P$ must now identify $c$ to $Q$ to distinguish it from other unoccupied (in the initial shape) neighbors that $Q$ may have.   
   We do not see how  to
   do that by communication alone. Instead, $P$ expands into $c$. It later sends messages from both $c$ and $b$ to $Q$, received in $Q$ on some ports $p, p'$ of node $a$. Suppose, w.l.o.g. that
the message from $c$ was received on port $p$. $Q$ then replies to each message, specifying the port number it was received on. Next, particle $P$ messages $Q$ from node $c$, telling $Q$ that $Q$ is the predecessor of $P$ with respect to the boundary (of the original shape) corresponding to unoccupied (originally) node $c$, and that the original node of $P$ is on port $p'$ of $a$.
Finally, $P$ contracts back to its original node $b$.
 $P$ performs a similar process for $b$'s successor node w.r.t.\ $B$ and for other boundaries $P$ lies on.

Once $P$ performs this routine for its predecessor and successor nodes along all boundaries that $P$ lies on, and in turn receives messages from those nodes, $P$ moves to the next phase. Thus, at the end of this phase, each particle $P$ knows, for each boundary $B$ it lies on, whether its successor and predecessor nodes along $B$ share the same chirality as $P$ or not.

\item \textit{Chirality agreement on each boundary and boundary identification:} Each particle $P$ participates in $\mathtt{Boundary-Detection}$ on each boundary $B$ it lies on. Once $P$ identifies each of the boundaries it lies on as inner or outer, $P$ moves to the next phase.

Simultaneously, each particle participates in a chirality alignment protocol for each of its boundaries. Define a \emph{common chirality segment} for a boundary $B$ as a contiguous sequence of occupied nodes on $B$ having the same chirality with a unique \emph{head node} and \emph{tail node} denoting those nodes in the segment without a predecessor or a successor, respectively. Consider two common chirality segments $s$ and $s'$ with adjacent tails $b$ and $c$ occupied by particles $P$ and $Q$, respectively. $P$ and $Q$ ``fight" each other and the ``losing" segment takes on the chirality of the ``winning" segment. Furthermore, all particles in the losing segment restart $\mathtt{Boundary-Detection}$. Recall that $b$ and $c$ have a unique common empty node $d$ for each boundary they share.\footnote{Note that $b$ and $c$ may share two common boundaries, both of which may be the outer boundary or a common inner boundary. However, during this phase, $b$ and $c$ do not know that the boundaries they share are the same boundary and hence $b$ and $c$ treat each boundary shared as a different boundary.} $P$ and $Q$ fight for this boundary by attempting to expand into $d$. Without loss of generality, let $P$ expand into $d$. $P$ then informs $Q$ to change its chirality w.r.t.\ $B$ and restart $\mathtt{Boundary-Detection}$, and then $P$ contracts back into $b$. $Q$ subsequently informs nodes in $s'$ to do the same.

\item \textit{Overall chirality agreement:} In this phase, all particles eventually take on the chirality of the outer boundary particles and signal this by setting flag $chir\_ready$ to $true$. An occupied node $b$ with $chir\_ready=false$ is a \textit{frontier node} if $b$ is either an outer boundary node, an inner boundary node which received an $align-chir$ message, or a node with at least two neighboring particles having $chir\_ready = true$, which are adjacent to each other. Consider a frontier node $b$ occupied by particle $P$. $P$ participates in a two stage process before setting $chir\_ready = true$ and moving to the next phase. Stage one consists of $P$ deciding its chirality, based on the type of frontier node $b$ is.

If $b$ is an outer boundary node with chirality $C$, then $P$ sets $C$ as the final chirality.  

Else if $b$ is an inner boundary node which received an $align-chir$ message, then $P$ sets its final chirality accordingly.\footnote{Suppose $b$ has chirality $C$ and received a $align-char$ message from its predecessor along a boundary $B$. The originating message stores the info on whether the final chirality is $C$ or not and $b$ updates its chirality accordingly.} 

Else $b$ has at least two adjacent occupied nodes with $chir\_ready = true$. When $b$ first detects it is a frontier node of this sort, let $P$ record the set of such nodes in $\mathcal{S}$. $P$ sends an $ask$ message to each node $c$ in $\mathcal{S}$, which in turn replies with $predecessor$, $successor$, or $both$ depending on whether $c$ is a predecessor node, successor node, or both w.r.t.\ $\mathcal{S}$ w.r.t.\ the boundary created by the absence of $b$.\footnote{Suppose node $c$ receives an $ask$ message from node $b$. There are either one or two neighboring nodes mutually adjacent to both $b$ and $c$ that have $chir\_ready = true$. If there are two neighboring nodes, $c$ replies with $both$. If there is one neighboring node $a$, then $c$ calculates the following. If $b$ was unoccupied and thus part of a boundary $B$ for $b$ and $a$, $b$ determines it is $a$'s predecessor or successor node w.r.t.\ $B$ and replies with $predecessor$ or $successor$, respectively.} Once $P$ receives replies from all nodes in $\mathcal{S}$, it checks if any new nodes should be added to $\mathcal{S}$. If so, $P$ updates $\mathcal{S}$ and sends new $ask$ messages to all nodes in $\mathcal{S}$. If not and $|\mathcal{S}| < 6$, then there exists at least one sequence of adjacent nodes in $\mathcal{S}$ such that at one end of the sequence, one node returned $predecessor$ and at the other end the other node returned $successor$. $P$ updates its chirality such that it agrees with the orientation of ports increasing from $successor$ to $predecessor$ through this sequence of ports. 

If $|\mathcal{S}| = 6$, $P$ chooses one node $c$ and sends it an $ask-special$ message. Again, if $b$ was unoccupied, then $b$ would be part of some boundary $B$ that $c$ lies on. Let $c$'s successor node w.r.t.\ $B$ be $d$. Now $c$ informs $d$ to send a $reply-special$ message to $b$.\footnote{It is possible for $d$ to uniquely identify the port leading to $b$. If $d$ received the message from $c$ through port $x$, then $d$ sends the $reply-special$ message through port $x+1 \mod 6$.} Once $b$ receives the $reply-special$ message from $d$, $P$ updates it chirality such that it agrees with the orientation of ports increasing from $d$ to $c$. 

Stage two consists of $P$ propagating any $align-char$ message it received to inner boundaries it lies on, if any, and subsequently setting $chir\_ready = true$. If $P$ received an $align-char$ message from its predecessor node on boundary $B$, it sends the same message to its successor node on $B$. $P$ also sends a $chir-align$ message to successor nodes on any other inner boundaries $P$ lies on. Once $P$ has sent and received a $chir-align$ message on every inner boundary it lies on, $P$ sets $chir\_ready = true$.

\item \textit{Termination detection:} For each particle $P$ with $chir\_ready = true$, if all adjacent particles have $chir\_ready = true$, then $P$ terminates the procedure.
\end{enumerate}

\begin{theorem}\label{the:chir-agreement}
Procedure $\mathtt{Chirality-Agreement}$, run by $n$ contracted particles forming a connected shape, terminates in $O(L_{\max}^2)$ rounds, where $L_{\max}$ is the length of the largest boundary in the shape, resulting in all particles having common chirality and retaining the original shape.
\end{theorem}

\begin{proof}
In order to show that all particles agree on chirality at the end of the procedure, we make use of two lemmas.

\begin{lemma}\label{lem:chir-frontier-node}
All particles in the shape eventually become frontier nodes.
\end{lemma}

\begin{proof}
Every particle $P$ is either an outer boundary particle, an inner boundary particle, or a non-boundary particle. By default, all outer boundary particles are frontier nodes. It can be seen by induction that every non-boundary particle eventually satisfies the condition that at least two of its neighboring particles which are adjacent to each other have $chir\_ready = true$. Every inner boundary particle either satisfies the previous condition or receives a $chir-align$ message. Thus every node eventually becomes a frontier node.
\end{proof}

\begin{lemma}\label{lem:chir-terminate}
Every frontier node terminates the procedure with the chirality agreed upon by the outer boundary.
\end{lemma}

\begin{proof}
Recall that the outer boundary is unique and all particles on it agree to the same chirality $C$ once phase 3 is complete. $C$ is the only chirality propagated to other particles in the system, and so once the procedure terminates all particles have chirality $C$. 
\end{proof}

It is clear from  Lemmas~\ref{lem:chir-frontier-node} and~\ref{lem:chir-terminate} that the procedure terminates with all particles agreeing on the same chirality. 

The running time of the procedure is the sum of the time taken in each phase. It is easy to see that phase 1 and 2 take $O(1)$ rounds each to complete. Phase 3 is the sum of the time it takes nodes on a boundary to agree on a chirality and the time it takes to complete $\mathtt{Boundary-Detection}$ on that boundary. For a boundary $B$, each particle $P$ on $B$ may undergo $O(L_{\max})$ resets of its chirality, with each reset taking at most $O(L_{\max})$ rounds to propagate to all nodes in $P$'s common chirality segment. Thus it takes $O(L_{\max}^2)$ rounds for particles on a boundary to agree on the same chirality. From Theorem~\ref{the:boundary-detection}, it takes an additional $O(L_{\max}^2)$ rounds to successfully run $\mathtt{Boundary-Detection}$ on any boundary. Phase 4 consists of transmitting the chirality of the outer boundary particles to all particles within the shape and takes no more than $O(n)$ rounds. Finally, phase 5 is completed in an additional $O(n)$ rounds. Recall that for any connected shape in the triangular grid, $L_{\max} = \Omega( \sqrt{n})$. Thus the total running time of the procedure is $O(L_{\max}^2)$ rounds.

Finally, we show that the original shape is retained at the end of the procedure. Notice that particles only expand and contract in phases 2 and 3. In both phases, a particle $P$ occupying node $b$ performs an expansion to some other node but always subsequently contracts back into $b$. Thus the original shape is unchanged at the end of the procedure.
\end{proof}

%----------------------------------------------------------------------------------------------

\section{Conclusion and Future Work}
\label{sec:conc}
% !TEX root = main.tex
 %please leave the above line untouched. Thanks, Billy.

In this paper, we developed algorithm $\texttt{Leader-Election-By-Moving}$, that when run by any contracted configuration of particles with common chirality, elects a unique leader deterministically. Subsequently, procedure $\texttt{Chirality-Agreement}$ was presented, which could be run by any contracted configuration of particles to agree on chirality deterministically. Thus any contracted configuration of particles, regardless of whether all particles have common chirality or holes are present, can run the procedure and then the algorithm in order to elect a unique leader deterministically. 

The results of this paper leave several lines of research open.
First, the algorithms here require the particles to move for leader election and for chirality agreement. Is it possible to solve either problem deterministically in the given setting without requiring particles to move? Second, can one reduce the running time or provide a matching lower bound?

%----------------------------------------------------------------------------------------------

%\bibliographystyle{abbrv}
\bibliographystyle{plainurl} %For LIPICS
\bibliography{references}

\end{document}